\documentclass[12pt,journal,onecolumn, draftclsnofoot]{IEEEtran}
\usepackage{enumerate}   
\usepackage[dvips]{graphicx}
\usepackage{epsfig}
\usepackage[cmex10]{amsmath}
\usepackage{amssymb}
\usepackage{amsthm}
\usepackage{amsfonts}
\usepackage{dsfont}
\usepackage{bm}
\usepackage{cite}
\usepackage[svgnames]{xcolor} 
\usepackage{pstricks,pst-node,pst-plot,pstricks-add}
\usepackage{algpseudocode}
\usepackage{stmaryrd}
\usepackage{algorithm}
\usepackage{nccmath}
\usepackage[OT1]{fontenc} 
\usepackage{cite}
\usepackage{multirow}
\usepackage[mathscr]{euscript}
\usepackage{subcaption}
\usepackage{caption}
\usepackage{setspace}
\usepackage{epsfig}
\usepackage{pst-grad} 
\usepackage{pst-plot} 
\usepackage[space]{grffile} 
\usepackage{etoolbox} 
\makeatletter 
\patchcmd\Gread@eps{\@inputcheck#1 }{\@inputcheck"#1"\relax}{}{}
\makeatother

		{\end{bmatrix}\end{medsize}}%
\graphicspath{{figs/}}
\input{niesen.def}

\begin{document}
\title{Multi-Party Proof Generation in QAP-based zk-SNARKs}

\author{Ali~Rahimi, Mohammad~Ali~Maddah-Ali%
	\\{Department of Electrical Engineering, Sharif University of Technology}%
}

\maketitle

\begin{abstract} 
Zero-knowledge succinct non-interactive argument of knowledge (zkSNARK)  allows a party, known as the prover,  to convince another party, known as the  verifier,  that he knows a private value $v$, without revealing it, such that $F(u,v)=y$ for some function $F$ and public values $u$ and $y$. There are various versions of zk-SNARK,  among them, Quadratic Arithmetic Program (QAP)-based  zk-SNARK has been widely used in practice, specially in Blockchain technology. This is attributed to two desirable features; its fixed-size proof and the very light computation load of the verifier. However, the computation load of the prover in QAP-based zkSNARKs, is very heavy, even-though it is designed to be  very efficient. This load can be beyond the prover's computation power to handle, and has to be offloaded to some external servers. In the existing offloading solutions, either (i)  the load of computation, offloaded to each sever, is a fraction of the prover's primary computation (e.g.,  DZIK), however the servers need to be trusted, (ii)  the servers are not required to be trusted, but the computation complexity imposed to each one is the same as the prover's primary computation (e.g.,  Trinocchio). In this paper, we present a scheme, which has the benefits of both solutions. In particular, we propose a secure multi-party proof generation algorithm where the prover can delegate its task to $N $ servers, where (i) even if a group of $T \in \mathbb{N}$ servers, $T\le N$,  collude, they cannot gain any information about the secret value $v$, (ii) the computation complexity of each server is less than $1/(N-T)$ of the prover's primary computation.  The design is such that we don't lose the efficiency of the prover's algorithm in the process of delegating the tasks to external servers. 
%
\end{abstract}


\section{Introduction}
\label{sec:introduction}
Zero knowledge proofs are powerful cryptographic tools by which a person (prover) can convince another person (verifier) that an assertion about some secret information is true, without revealing the secret information itself or revealing any other information beyond the fact that the assertion is true. Consider, for example, the assertion that "I know a secret number that is quadratic non-residue mod $m$''. The  prover can convince the verifier that he knows such a number without giving additional knowledge to the verifier. \par
The zero knowledge proofs were initially introduced in 1980s, in \cite{goldwasser1989knowledge} and \cite{babai1985trading}. After that, a lot of efforts have been dedicated to make it more efficient. For example, Kilian \cite{kilian1992note} has introduced the \emph{succinct} interactive zero knowledge proof system, where the load of communication between the prover and the verifier can be less than the size of the corresponding arithmetic circuit, a diagram which representing the original computation. Micali \cite{micali1994cs} has developed zero-knowledge succinct \emph{non-interactive} argument of knowledge (zkSNARK), which is a type of zero knowledge proof systems in which the prover just sends one message to the verifier. For a survey about zero knowledge proof systems see \cite{li2014survey} and \cite{mohr2007survey}. \par

zkSNARKs have been extensively explored in the literature \cite{micali1994cs, setty2013resolving, bitansky2013succinct, gabizon2019plonk, giacomelli2016zkboo, maller2019sonic, setty2019spartan, bowe2019halo, wahby2018doubly, bootle2018arya, ben2017scalable, xie2019libra, ames2017ligero, wahby2017full, bunz2018bulletproofs, ben2018scalable, ben2019aurora}. Also zkSNARKs have been used in many applications, for example in authentication systems \cite{kurmi2015survey}, construction of various types of cryptographic protocols \cite{groth2017snarky}, privacy preserving crypto-currencies \cite{sasson2014zerocash}, smart contracts \cite{kosba2016hawk, juels2016ring, li2020phantom}, miscellaneous applications in Blockchain technology\cite{xu2017enabling, lu2018zebralancer, kerber2019ouroboros, garoffolo2020zendoo}, verifiable outsourcing of the computation \cite{shan2018practical, yu2017survey}, and many other areas \cite{destefano2020snnzksnark}. In addition, various toolboxes have been created to implement zkSNARKs, e.g., \cite{libsnark, bellman, websnark, eberhardt2018zokrates}. \par

In \cite{setty2019spartan}, an efficient zkSNARK is introduced which is based on a specific encoding of the arithmetic circuit to a  polynomial form called quadratic arithmetic program (QAP). QAP-based zkSNARKs have several properties that make them attractive in practice. In particular, the  size of the proof is always \emph{constant}, e.g., 8 elements in \cite{parno2013pinocchio}, and 3 elements in \cite{groth2016size} regardless of the size of  the arithmetic circuit. In addition, the verifier's work (the number of addition and multiplication operations) is of the order $O(| \mathcal{I}_{io} |)$, where $| \mathcal{I}_{io} |$ is the aggregated number of public inputs and outputs of the arithmetic circuit (see section \ref{sec:Background}). However, the prover's work in QAP-based zkSNARKs is dominated with $O(n \, \log(n))$ computation cost, where $n$ is the number of multiplication gates in the arithmetic circuit. When the arithmetic circuit is large, e.g., with more than several billions of gates, a prover with a limited computing resource cannot generate the proof himself, and he may need to offload this computation to some other servers. \par

Several papers have investigated the problem of outsourcing the task of generating proof. 
\begin{itemize}
\item \textbf{Nectar \cite{covaci2018nectar}}: Nectar which is a smart contract protocol, uses zkSNARKs for verification of the correct execution of the smart contracts. In Nectar, the prover delegates his task to a powerful  trusted worker, and sends all of its secret inputs to that worker to generate the proof. The disadvantage of this system is the need to trust the external server. In addition, that server needs to be powerful enough to handle the computation. 
\item \textbf{DIZK \cite{wu2018dizk}}: This work proposes an algorithm for delegation of the prover task to several trusted machines by using a Map-Reduce framework. The advantage is that each server is responsible to execute part of the prover task, which is assigned to it based on its processing and storage resource. However, if some of the servers are untrusted, DIZK framework cannot be used.
\item \textbf{SPARKs \cite{ephraim2020sparks}}: This work breaks the computation task into a sequence of smaller sub-tasks, and delegates the proof generation for the correctness of each sub-task to one server. The weakness is that the servers must be trusted. On top of that, SPARKs deal with the computation task in detail to be able to split it into sub-tasks.
\item \textbf{Trinocchio \cite{schoenmakers2016trinocchio}}: In Trinocchio, we run BGW multi party computation (MPC) \cite{ben2019completeness} to offload the computation to several servers. The disadvantage is that the load of the task assigned to each server is the same as that of the original computation. 
\end{itemize}

To offload computation task on private data to some external untrusted nodes, a fundamental approach is multi party computation.  The multi party computation was initially introduced in the early 1980s by Yao \cite{yao1982protocols}, and has been followed in many works \cite{yao1986generate, goldreich2019play, ben2019completeness, crepeau1995committed, yu2018lagrange }. For a survey about multi party computation see \cite{evans2017pragmatic}. Multi party computation has been used in many areas such as machine learning \cite{chen2019secure}, secure voting \cite{gang2008electronic}, securing database \cite{archer2018keys} and Blockchain technology \cite{gao2019bfr}.

In this paper, our objective is to design an offloading mechanism that has two properties: 
\begin{enumerate}
\item The load of computation per server is a faction of the full load of generating proofs. The reason is clear: the load of computation for prover task is beyond what one server can affords. 

\item The servers are not required to be trusted. In particular, we assume a subset of size $T \in \mathbb{N}$ of the servers may collude to gain information about the input of the prover.
\end{enumerate}

Current solutions have only one of the above properties. In particular, Trinocchio \cite{schoenmakers2016trinocchio} works even if some of the servers are curious and colludy. However, the load of computation per server is the same or even more than the load of prover task. 
In DIZK~\cite{wu2018dizk}, on the other hand, the load of computation per server can be a fraction of the load of the prover task. However, it does assume that all the servers are trusted. Our approach is based on ideas from multiparty computation.  However, we cannot use an off-the-shelf multiparty computation scheme and apply it to the computation task of the prover. The reason is that the prover task has been hand designed to be very efficient. If we apply an MPC scheme blindly, we lose the efficiency in computation, and the computation task of each server becomes even more that the original computation. 
In this paper, we design a multi-party scheme for the prover task that works with $N$ servers, such that, even if $T$ of them collude for some $T < N$, they gain no information about the secret inputs (the second property). In addition, the computation load of each server is $\frac{n}{N-T} \log \left(\frac{n}{N-T}\right)$ (the first property).

%

The rest of the paper is organized as follows. In the section \ref{sec:Background}, we review some of the backgrounds including QAP-based zkSNARKs. In the section \ref{Preliminaries2}, we review fast Fourier transform ($\mathscr{FFT}$) and some secret sharing schemes. In the section \ref{sec:Proposed-Scheme}, we explain the main challenge and then we detail the proposed scheme. Section \ref{Discussion and conclusion} is dedicated to discussion and conclusion.

\textbf{Notation.} We denote vectors by lowercase bold letters such as $\mathbf{x}$. The notation $\mathbf{x} = [x_i]_{i=1}^{n}$ means $\mathbf{x}$ is a vector of length $n$ and $x_i$ is its $i$th coordinate. There are some vectors in this paper that has a large mathematical symbol, representing their construction path. To show the $i$th coordinate of these vectors we use the notation $(.)_{i}$. For example $\left ( \mathscr{FFT}_{\mathcal{S'}}^{-1} \left (\mathbf{u} \left (\alpha_\theta \right) \right) \right )_{i}$ is the $i$th coordinate of the vector $\mathscr{FFT}_{\mathcal{S'}}^{-1}(\mathbf{u}(\alpha_\theta))$.

We denote matrices by bold uppercase letters, e.g. $\mathbf{X}$. We denote sets by uppercase calligraphy letters and use $\{.\}$ to show the elements, for example in $\mathcal{X} = \{x_1,x_2,x_3\}$, $\mathcal{X}$ is a set containing elements $x_1$, $x_2$, and $x_3$. We use double bracket $\llbracket .\rrbracket$ for encryption. For more detail see \emph{Cryptographic operations} in Section \ref{preliminaries}.

\section{Background on QAP-based zkSNARK}
\label{sec:Background}
\subsection{The story of zkSNARK}
Suppose that there is a globally known function $y=F(u,v)$ consisting of only multiplication and addition operations. The prover is a person that has calculated this function with inputs $u$ and $v$, and has obtained $y$. The input $u$ and the output $y$ are publicly available. The prover wants to convince the verifier that he know $v$ such that $y=F(u,v)$. However, one of the main constraints is that $v$ is a private parameter, and the prover doesn’t want to reveal it to the verifier. The second constraint is that the verifier wants to verify this computation with negligible load of computation and communication.

A non-interactive zero-knowledge proof system allows the prover to make a string $\pi$, called \emph{proof}, that if the prover sends it along with public input $u$ and output $y$ to the verifier, the verifier will be convinced that the prover knows a $v$ as the input of $F$ such that the calculation of $y=F(u,v)$ has been done correctly without obtaining any other information about $v$. 

In order to generate and verify proofs, a process must be performed in advance based on the structure of function $F$. This process is called the \emph{setup phase}. A third party, often called as the trusted party, runs the setup phase and generates two public parameters \emph{Evaluation Key} ($\mathcal{EK}$) and \emph{Verification Key} ($\mathcal{VK}$). $\mathcal{EK}$ and $\mathcal{VK}$ depend on the structure of the function $F$, and are independent of $y$, $u$, $v$. Then the prover generates proof $\pi$ using $\mathcal{EK}$ and the result of the calculation $y = F(u,v)$. The verifier verifies the prover claim using $\mathcal{VK}$ and the proof $\pi$. The size of the proof and the computation load of verifying should be negligible. 

We note that the setup phase is a one-time process. In other words, $\mathcal{EK}$ and $\mathcal{VK}$ can be used many times as long as the function $F$ remains the same. As a result, the computation cost of the setup phase amortizes over many zkSNARK sessions about $F$ by different provers. It is worth mention that in the setup phase, the trusted party uses some intermediate parameters to develop $\mathcal{EK}$ and $\mathcal{VK}$. These parameters are used only once, and must be deleted after that; otherwise if someone has access to these parameters can cheat and generate counterfeit proofs (See Fig. \ref{snark_components}). 


As an example, in the Zcash Blockchain which uses zkSNARK to support anonymous transactions, the setup phase had been run before the network started, and two sets $\mathcal{EK}$ and $\mathcal{VK}$ made available to everyone. Whenever someone wants to make an anonymous transaction, he should generate a proof using $\mathcal{EK}$ (to prove he has enough money, etc.), then the miners in the network verify the proof using $\mathcal{VK}$. For more details see \cite{sasson2014zerocash}.

zkSNARKs have some properties that are informally mentioned below:
\begin{itemize}
\item \textbf{Zero-knowledge}: Verifier obtains no information about $v$ beyond the fact that $y=F(u,v)$.
\item \textbf{Succinctness}: Size of the proof $\pi$ is constant, no matter the size of $F$. This feature makes zkSNARK a good tool in practice (e.g. cloud computing).
\item \textbf{Non-interactive}: Verifier doesn't send anything to the prover.
\item \textbf{Publicly verifiable}: This property allows many people to check the proof and it is useful for applications such as Blockchain.
\item \textbf{Correctness}: If zkSNARK is executed honestly and a proof is generated honest, verifier(s) will always detect it correctly.
\item \textbf{Knowledge soundness}: Polynomial-time adversary who doesn't know some $v$ that holds in $y = F(u,v)$, can not generate a valid proof.
\end{itemize}

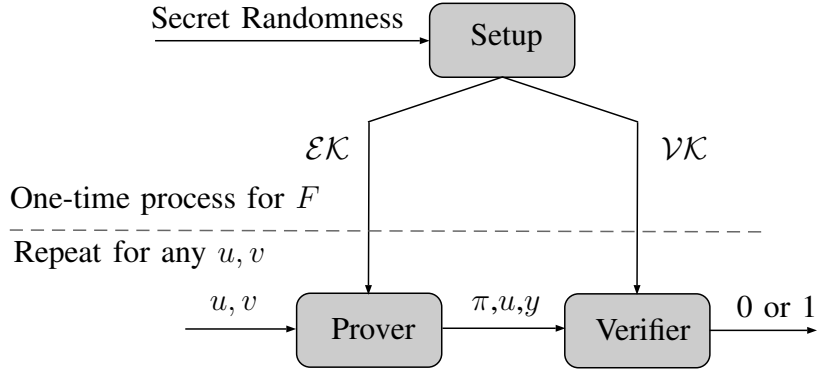
\begin{figure}[!htbp]
	\centering
%
\psscalebox{1.0 1.0} 
{
\begin{pspicture}(0,-2.4341586)(10.666204,2.4341586)
\definecolor{colour0}{rgb}{0.8,0.8,0.8}
\definecolor{colour1}{rgb}{0.4,0.4,0.4}
\psframe[linecolor=black, linewidth=0.02, fillstyle=solid,fillcolor=colour0, dimen=outer, framearc=0.3701811](7.5209184,2.4341586)(5.571712,1.430984)
\psframe[linecolor=black, linewidth=0.02, fillstyle=solid,fillcolor=colour0, dimen=outer, framearc=0.3701811](5.752664,-1.4231112)(3.8034577,-2.4262857)
\psline[linecolor=black, linewidth=0.02](6.5535297,1.4453795)(4.76153,0.8453795)
\psline[linecolor=black, linewidth=0.02](6.54553,1.4453795)(8.35353,0.8453795)
\rput[bl](3.9311967,0.35604623){$\mathcal{EK}$}
\rput[bl](6.126315,1.7825713){Setup}
\rput[bl](1.8777628,2.0931256){Secret Randomness}
\rput[bl](4.278061,-2.0446985){Prover}
\psframe[linecolor=black, linewidth=0.02, fillstyle=solid,fillcolor=colour0, dimen=outer, framearc=0.3701811](9.317743,-1.4309841)(7.368537,-2.4341588)
\rput[bl](7.788141,-2.0605712){Verifier}
\rput[bl](8.660197,0.35604623){$\mathcal{VK}$}
\psline[linecolor=black, linewidth=0.02, arrowsize=0.05291667cm 2.0,arrowlength=1.4,arrowinset=0.0]{->}(4.776123,0.84122247)(4.768123,-1.4547775)
\psline[linecolor=colour1, linewidth=0.02, linestyle=dashed, dash=0.17638889cm 0.10583334cm](0.0,-0.59199375)(9.989363,-0.6159938)
\rput[bl](0.0036502185,-0.39436337){One-time process for $F$}
\rput[bl](0.051782332,-1.111503){Repeat for any $u,v$}
\rput[bl](6.1392603,-1.7423352){$\pi$,$u$,$y$}
\psline[linecolor=black, linewidth=0.02, arrowsize=0.05291667cm 2.0,arrowlength=1.4,arrowinset=0.0]{->}(2.3331509,-1.9083599)(3.8066325,-1.9083599)
\rput[bl](2.6584842,-1.7423352){$u,v$}
\psline[linecolor=black, linewidth=0.02, arrowsize=0.05291667cm 2.0,arrowlength=1.4,arrowinset=0.0]{->}(9.303987,-1.9268149)(10.731246,-1.9268149)
\rput[bl](9.664114,-1.7523352){0 or 1}
\psline[linecolor=black, linewidth=0.02, arrowsize=0.05291667cm 2.0,arrowlength=1.4,arrowinset=0.0]{->}(1.9317541,1.9268148)(5.5812354,1.9268148)
\psline[linecolor=black, linewidth=0.02, arrowsize=0.05291667cm 2.0,arrowlength=1.4,arrowinset=0.0]{->}(5.7521987,-1.898836)(7.3875847,-1.898836)
\psline[linecolor=black, linewidth=0.02, arrowsize=0.05291667cm 2.0,arrowlength=1.4,arrowinset=0.0]{->}(8.346682,0.8412227)(8.338682,-1.4547774)
\end{pspicture}
}
	\caption{Components of zkSNARK system ~\cite{wu2018dizk}. Note that $F$ and its corresponding Quadratic Arithmetic Program are available to everyone.}
	\label{snark_components}
\end{figure}

\subsection{Main components of zkSNARK}
\label{preliminaries}
\begin{itemize}
\item \textbf{Arithmetic circuit}: It's a diagram, consists of wires and multiplication and addition gates, to represent the function $F$ and its intermediate calculations. For example, Fig. \ref{arithmetic_circuit_a} corresponding to the function $(c_1+c_2)c_3^2$. As it turns out, this diagram represents not only the function, but also represents the process of computation. For a computation to be correct, every computation in this graph must be correct. As you will see later, this structure would allow us to develop the proof of correct execution of the function $F$. The arithmetic circuit that corresponds to a function is not unique but a valid representation is enough for us. 


\item \textbf{Equivalent Quadratic Arithmetic Program (QAP)}: In this step, we represent the structure of the arithmetic circuit using some polynomials.  Suppose that the arithmetic circuit of $F$, has $n$ multiplication gates. We assume  $n$ is a power of 2. If it is not the case, we add some operations to the arithmetic circuit to make $n$ become a power of 2. To develop the corresponding polynomials we need to label the multiplication gates and wires of the circuit. Let $\omega$ be a primitive $n$th root of unity in $\mathbb{F}$, i.e., $\omega^n = 1$. We label the multiplication gates of the arithmetic circuit by the set $\mathcal{S} = \left \{1,\dots,\omega^{n-1} \right \}$ in an arbitrary order. 

Now consider the wires that are input of the arithmetic circuit and the wires that are output of the multiplication gates. We index them in an arbitrary order by the set $\{1,\dots,m\}$. As a convention, there is a wire in the arithmetic circuit that always carries 1. We assign index $i = 0$ to that wire as well.  As you can see, we neither  label the addition gates nor index their output  wires. Later, we will explain how to treat those.


Now we are ready to represent the structure of the arithmetic circuit in polynomials. For each wire, indexed by $i$, $i \in \{0,\dots,m\}$, we define three polynomials of the degree at most $n-1$, denoted by $L_{i}(x)$, $R_{i}(x)$ and $O_{i}(x)$ as follows: 

$L_{i}(x) = 
\begin{cases}
    1,& \text{if } x = \text{label of a gate, and the } i \text{th wire} \text{ is the left input of that gate}\\
    0,& \text{if } x = \text{label of a gate, and the } i \text{th wire} \text{ isn't the left input of that gate}\\
    \sim,& \text{otherwise}
\end{cases}$

The symbol $\sim$ means the value of the polynomial does not matter at this point.

$R_{i}(x) = 
\begin{cases}
    1,& \text{if } x = \text{label of a gate, and the } i \text{th wire} \text{ is the right input of that gate}\\
    0,& \text{if } x = \text{label of a gate, and the } i \text{th wire} \text{ isn't the right input of that gate}\\
    \sim,& \text{otherwise}
\end{cases}$

$O_{i}(x) = 
\begin{cases}
    1,& \text{if } x = \text{label of a gate, and the } i \text{th wire} \text{ is the output of that gate}\\
    0,& \text{if } x = \text{label of a gate, and the } i \text{th wire} \text{ isn't the output of that gate}\\
    \sim,& \text{otherwise}
\end{cases}$

These polynomials can be developed simply by Lagrange interpolation.  

Recall that we don't label the output wires of the addition gates. To cover the addition gates and their outputs, in the above definitions, we extend the notation of being a \emph{right input} or \emph{left input} of a multiplication gate as follows:  If an indexed wire goes through one or more addition gates, and eventually becomes the right (left) input of a multiplication gate, then we also consider that indexed wire as a right (left) input of that multiplication gate. By going through an addition gate, we mean it is an input of that addition gate. For example in the Fig. \ref{arithmetic_circuit}, we say wire 1 and wire 2 both are the left inputs of multiplication gate.

We also define polynomial, $T(x) \triangleq x^n - 1$, called \emph{target polynomial}, which is of the degree $n$, and is divisible by the label of each multiplication gate. In other words,  if $x \in \mathcal{S}$, we have $x = \omega^j$ for some $0\leq j < n$, so $T(x) = T(\omega^j) = \omega^{jn} - 1 = 0$.

A QAP $\mathcal{Q} \triangleq \left \{T(x), \left \{L_{i}(x)\right \}, \left \{R_{i}(x)\right \}, \left \{O_{i}(x)\right \}\right \}$ over the finite field $\mathbb{F}$ is the set of target polynomial $T(x)$ and three sets of $m+1$ polynomials. We note that QAP of an arithmetic circuit  completely describes the structure of that circuit.

\begin{figure}
     \centering
     \begin{subfigure}[b]{0.5\textwidth}
         \centering
%
\psscalebox{1.0 1.0} 
{
\begin{pspicture}(0,-2.1466668)(5.464805,2.1466668)
\definecolor{colour0}{rgb}{0.8,0.8,0.8}
\pspolygon[linecolor=black, linewidth=0.01, fillstyle=solid,fillcolor=colour0](5.4597588,1.0285784)(4.7166476,1.0285784)(4.4464254,0.83302283)(4.4464254,0.62413394)(4.6997585,0.41968948)(5.454129,0.42413393)
\pscircle[linecolor=black, linewidth=0.01, dimen=outer, doubleline=true, doublesep=0.02](0.8139349,-0.10812661){0.08888889}
\pspolygon[linecolor=black, linewidth=0.01, fillstyle=solid,fillcolor=colour0](0.005046041,0.19409561)(0.74815714,0.19409561)(1.0183793,-0.0014599424)(1.0183793,-0.21034883)(0.76504606,-0.41479328)(0.010675671,-0.41034883)
\rput[bl](1.7234287,1.0333334){wire 4}
\psellipse[linecolor=black, linewidth=0.02, dimen=outer](2.337036,-0.10565346)(0.1904762,0.1904762)
\psline[linecolor=black, linewidth=0.02](2.2536385,-0.033066727)(2.420434,-0.1998621)
\psline[linecolor=black, linewidth=0.02](2.2536383,-0.199862)(2.4204338,-0.033066574)
\psellipse[linecolor=black, linewidth=0.02, dimen=outer](1.7222613,-0.9400679)(0.1904762,0.1904762)
\psline[linecolor=black, linewidth=0.02](1.7222614,-0.8221257)(1.7222614,-1.05801)
\psline[linecolor=black, linewidth=0.02](1.6043192,-0.94006765)(1.8402035,-0.94006765)
\rput[bl](0.22559094,-1.3693243){wire 1}
\rput[bl](1.3488549,0.30033088){wire 3}
\rput[bl](2.1493971,-1.3693243){wire 2}
\psline[linecolor=black, linewidth=0.02, arrowsize=0.05291667cm 2.0,arrowlength=1.4,arrowinset=0.0]{->}(1.1538793,-1.6249373)(1.6212666,-1.1024147)
\psline[linecolor=black, linewidth=0.02, arrowsize=0.05291667cm 2.0,arrowlength=1.4,arrowinset=0.0]{->}(1.7509964,-0.75207204)(2.2050505,-0.25081077)
\psline[linecolor=black, linewidth=0.02, arrowsize=0.05291667cm 2.0,arrowlength=1.4,arrowinset=0.0]{->}(4.156762,-1.6)(2.4702756,-0.24324322)
\psline[linecolor=black, linewidth=0.02, arrowsize=0.05291667cm 2.0,arrowlength=1.4,arrowinset=0.0]{->}(2.9302793,0.9027028)(2.9299116,1.4688288)
\psline[linecolor=black, linewidth=0.02, arrowsize=0.05291667cm 2.0,arrowlength=1.4,arrowinset=0.0]{->}(2.337663,0.08792796)(2.791717,0.58918923)
\psellipse[linecolor=black, linewidth=0.02, dimen=outer](2.937036,0.7076799)(0.1904762,0.1904762)
\psline[linecolor=black, linewidth=0.02](2.8536384,0.7802666)(3.020434,0.61347127)
\psline[linecolor=black, linewidth=0.02](2.8536384,0.6134713)(3.0204337,0.78026676)
\psline[linecolor=black, linewidth=0.02, arrowsize=0.05291667cm 2.0,arrowlength=1.4,arrowinset=0.0]{->}(4.1700954,-1.6)(3.0569422,0.5567568)
\rput[bl](3.6782036,-0.515991){wire 5}
\psframe[linecolor=black, linewidth=0.02, dimen=outer](2.476762,-1.6133333)(1.9434288,-2.1466665)
\rput[bl](2.0600955,-2.0006897){$c_2$}
\psframe[linecolor=black, linewidth=0.02, dimen=outer](1.3834288,-1.6133333)(0.85009545,-2.1466665)
\rput[bl](0.9667621,-2.0006897){$c_1$}
\psframe[linecolor=black, linewidth=0.02, dimen=outer](4.4367623,-1.5999999)(3.9034288,-2.1333332)
\rput[bl](4.015095,-1.9873562){$c_3$}
\psframe[linecolor=black, linewidth=0.02, dimen=outer](3.916762,2.1466668)(1.7112448,1.4666667)
\rput[bl](1.8345988,1.6147702){$(c_1 + c_2)c_3^2$}
\psline[linecolor=black, linewidth=0.02, arrowsize=0.05291667cm 2.0,arrowlength=1.4,arrowinset=0.0]{->}(2.2346,-1.6249373)(1.7672125,-1.1024147)
\rput[bl](4.8585167,0.6475257){$\omega$}
\psline[linecolor=black, linewidth=0.01, linestyle=dashed, dash=0.17638889cm 0.10583334cm, dotsize=0.07055555cm 2.0]{*-}(4.5700955,0.7291954)(3.114923,0.72229886)
\psline[linecolor=black, linewidth=0.01, linestyle=dashed, dash=0.17638889cm 0.10583334cm, dotsize=0.0706cm 3.53]{-*}(2.1080265,-0.09839077)(0.87354374,-0.10528733)
\rput[bl](0.19768165,-0.21563216){$1$}
\end{pspicture}
}
         \caption{Arithmetic circuit}
         \label{arithmetic_circuit_a}
     \end{subfigure}
     \hfill
     \begin{subfigure}[b]{0.4\textwidth} 
         \centering
%
\psscalebox{1.0 1.0} 
{
\begin{pspicture}(0,-1.9566667)(7.124865,1.9566667)
\psline[linecolor=black, linewidth=0.02](0.0,1.6093693)(7.1248646,1.6093693)
\psline[linecolor=black, linewidth=0.02](3.859793,1.6115602)(3.8464596,-1.1359575)
\rput[bl](0.20378374,-0.5071171){$L_4(x),R_4(x),O_4(x)$}
\rput[bl](0.20378374,-1.0537838){$L_5(x),R_5(x),O_5(x)$}
\rput[bl](0.20378374,0.033423424){$L_3(x),R_3(x),O_3(x)$}
\rput[bl](0.20378374,0.5872973){$L_2(x),R_2(x),O_2(x)$}
\rput[bl](0.20378374,1.1145045){$L_1(x),R_1(x),O_1(x)$}
\rput[bl](4.3845944,1.1036937){$1,0,0$}
\rput[bl](4.3745947,0.57648647){$1,0,0$}
\rput[bl](4.3745947,0.035945944){$0,0,1$}
\rput[bl](4.3845944,-1.0512613){$0,1,0$}
\rput[bl](4.3845944,-0.5045946){$0,0,0$}
\rput[bl](4.3645945,1.7366667){$x=1$}
\rput[bl](5.9645944,1.7366667){$x=\omega$}
\rput[bl](5.9979277,-0.5045946){$0,0,1$}
\rput[bl](5.9979277,-1.0512613){$0,1,0$}
\rput[bl](5.987928,0.035945944){$1,0,0$}
\rput[bl](5.987928,0.57648647){$0,0,0$}
\rput[bl](5.9979277,1.1036937){$0,0,0$}
\rput[bl](2.2336936,-1.9566667){$T(x) = (x-1)(x-\omega)$}
\end{pspicture}
}
         \caption{QAP}
         \label{arithmetic_circuit_b}
     \end{subfigure}
        \caption{The arithmetic circuit of the function $(c_1+c_2)c_3^2$ and the corresponding QAP.}
       \label{arithmetic_circuit}
\end{figure}

\item \textbf{Polynomial Representation of the Correctness of the Operations}: Recall that the advantage of the arithmetic circuit  of a function $F$ is that it represent all the intermediate operations in calculating $F$. Let us assume that we calculate $F$, and in this process, we also calculate the value that is carried by each indexed wire $i$ , denoted by $W_i$, $i \in \{0, \ldots, m\}$. For final results to be correct, we need the calculation in each multiplication gate to be correct. To verify that, one needs to verify $n$ operations one by one, which would be very difficult.  An interesting aspect of QAP is that we can use it to represent all of these operations with \emph{one} polynomial equation, as follows. A polynomial equation can easily be verified reasonably as well as will be explained later.  

Recall that all of the polynomials $\{L_i(x)\}$, $\{R_i(x)\}$ and $\{O_i(x)\}$ are of the degree $n-1$. We define the polynomials $L(x), R(x), O(x)$ of the degree $n-1$ and $P(x)$ of the degree $2n-2$ as,
	\begin{align} 
	L(x) \triangleq \sum_{i=0}^{m}W_iL_i(x), \: \:
	R(x) \triangleq \sum_{i=0}^{m}W_iR_i(x), \: \:
	O(x) \triangleq \sum_{i=0}^{m}W_iO_i(x),
	\nonumber
	\end{align}
	\begin{align}
	P(x) \triangleq L(x)R(x) - O(x),
	\nonumber
	\end{align}
	where $W_i$ is the value of the indexed wire $i$. An important observation is as follows. Let $ \omega^j$ be the label of multiplication gate $j$. Then, one can see that  $L(\omega^j)$ is equal to (the summation of) the values of wires that are left inputs of gate $j$. Similarly, $R(\omega^j)$ and $O(\omega^j)$  are equal to (the summation of) the values of the wires that are right inputs and the output wire of gate $j$, respectively. Thus for the calculation at gate $j$ to be correct,  we need to have  $L(\omega^j)R(\omega^j)=O(\omega^j)$. As the result, if the prover wants to prove that the calculation of the entire arithmetic circuit has been done correctly, it's sufficient to show $P(x)=0$, $\forall \ x \in \mathcal{S}$. Equivalently, it is sufficient to show that the target polynomial $T(x)$ divides $P(x)$. In other words, the prover needs to show that there is a polynomial $H(x)$ of the degree at most $n-2$ that $P(x)=T(x)H(x)$. 
	
%

The main idea behind zkSNARK is that (i) the prover finds polynomial $H(x)$ and then (ii) the verifier checks the equation $P(x)=T(x)H(x)$ in a point $x = s$ chosen uniformly at random from $\mathbb{F}$.  If the equation $P(x)=T(x)H(x)$ doesn't hold, the verifier will detect it with high probability. This is because two different polynomials of degree $2n-2$ can have equal values in at most $2n-2$ different points, and assuming $\left | \mathbb{F} \right | \gg 2n-2$, the probability that $s$ be one of those $2n-2$ points is $\frac{2n-2}{\left | \mathbb{F} \right |}$ that is negligible. The important note is that the prover should not know the value of $s$, otherwise, it can introduce invalid polynomials $L(x)$, $R(x)$, and $O(x)$, such that the identity $P(x)=T(x)H(x)$ holds only for $s$. This is why some cryptographic operations are needed to verify the equation for encrypted numbers. 

\item \textbf{Cryptographic operations}: In QAP based zkSNARK we rely on elliptic curve cryptography to protect the private data and soundness of the algorithm. 
Let $\mathbb{G}_1$ be an additive group,  developed based on an elliptic curve defined over the finite field $\mathbb{F}$, and $g_1 \in \mathbb{G}_1$ be a generator of this group. To encrypt a scalar $a \in \mathbb{F}$, we calculate  $g_1+g_1+\dots+g_1$ with $a$ appearance of $g_1$ in the additive group $\mathbb{G}_1$, and denote it as $a. g_1$ or $\llbracket a \rrbracket_1$.  We note that finding  number $a$ from $\llbracket a \rrbracket_1$ is computationally infeasible. In addition, different inputs lead to different outputs. Moreover, the encryption operation is linear, i.e.,  $\llbracket a + b \rrbracket_1 = \llbracket a \rrbracket_1+ \llbracket b \rrbracket_1$, for two integers $a$ and $b$ \cite{el2017guide}.

Consider three integers $a$, $b$, and $c$, and assume  that we only have access to $\llbracket a \rrbracket_1, \llbracket b \rrbracket_1, \llbracket c \rrbracket_1$. Let us assume that we aim to verify if $c=a+b$. This can be simply done by checking if $\llbracket c \rrbracket_1 = \llbracket a \rrbracket_1 + \llbracket b \rrbracket_1$.  Now let us assume that we want to check of $c = ab$. This is not straight-forward, and is done through the notation of \emph{pairing} $e$. 
 
Let $e: \mathbb{G}_1 \times \mathbb{G}_2 \rightarrow \mathbb{G_{T}}$ be a non-trivial bilinear map from two groups  $\mathbb{G}_1$ and $\mathbb{G}_2$ to a group $\mathbb{G_{T}}$, and $g_1,g_2$ be generators of $\mathbb{G}_1, \mathbb{G}_2$ respectively. It has three properties:
\begin{itemize}
\item $g_1,g_2 \neq 1 \Rightarrow e(g_1,g_2) \neq 1$,
\item $\forall a,b \in \mathbb{F}: e(\llbracket a \rrbracket_1, \llbracket b \rrbracket_2) = ab \, e(g_1,g_2)$, where $\llbracket b \rrbracket_2$ is $b g_2$,
\item $e$ is efficiently computable.
\end{itemize}

Now if we have $\llbracket a \rrbracket_1$, $\llbracket b \rrbracket_2$ and $\llbracket c \rrbracket_1$, the encrypted versions $a$, $b$ and $c$ respectively, we can check $c = a * b$ by checking the equation $e(\llbracket a \rrbracket_1, \llbracket b \rrbracket_2) = e(\llbracket c \rrbracket_1, \llbracket 1 \rrbracket_2)$. We note that $\mathbb{G}_1$ and $\mathbb{G}_2$ can be the same group, with $g_1$ as the generator. In that case, we check if $e(\llbracket a \rrbracket_1, \llbracket b \rrbracket_1) = e(\llbracket c \rrbracket_1, \llbracket 1 \rrbracket_1)$

As mentioned, we use double brackets $\llbracket . \rrbracket$ to show the encrypted version of the scalars, but vectors can also be represented in this notation, for example $\llbracket [x_1,x_2,x_3] \rrbracket_1 = \left [ \llbracket x_1 \rrbracket_1, \llbracket x_2 \rrbracket_1, \llbracket x_3 \rrbracket_1\right ]$. Using this notation, we also have $\llbracket a x_1 + b x_2 \rrbracket_1 = a \llbracket x_1 \rrbracket_1 + b \llbracket x_2 \rrbracket_1$. 

For more details see \cite{washington2008elliptic}.
\end{itemize}

\subsection{Three main algorithms in zkSNARK} \label{Groth_prover}
There exist various versions of zkSNARK. Here we focus on the version proposed by Groth in \cite{groth2016size} which is one of the most efficient and popular QAP-based zkSNARKs. However, the schemes proposed in this paper can be applied to other variations of QAP-based zkSNARKs. In the following, we review three algorithms included in zkSNARK, the setup phase algorithm which is done  only once by an entity called as the \emph{trusted party}, the prover algorithm which is done by the \emph{prover}, and the verifier algorithm which is done by the \emph{verifier}. 
\begin{itemize}
\item \textbf{Setup phase algorithm}: This algorithm takes the function $F$ and the security parameter $\kappa \in \mathbb{N}$ as the input, and outputs $\mathcal{EK}$ and $\mathcal{VK}$. The security parameter specifies the size of the finite field $\mathbb{F}$. If $\kappa$ is large, the algorithm is more secure, at the cost of increasing the computation load. 

The setup phase algorithm is presented in Algorithm \ref{Setup}. 
We note that in Line 3, some random parameters are chosen. These random parameters are used to develop  $\mathcal{EK}$ and $\mathcal{VK}$ and will be deleted at the end of set-up phase.  $\mathcal{I}_{io}$, in Line 4,  is the set of indices of the wires that are the public input or the output of the arithmetic circuit. $\mathcal{I}_{mid}$, in Line 5, is the set of indices of the wires that are not the public input nor the output of the arithmetic circuit. It is obvious that $\mathcal{I}_{io} \cup \mathcal{I}_{mid} = \{0,...,m \}$.

In Line 9 and Line 10, $\mathcal{EK}$ and $\mathcal{VK}$ are generated respectively. All values generated during the algorithm except $\mathcal{EK}$ and $\mathcal{VK}$ are known as \emph{toxic waste}, and must be deleted at the end of the algorithm for ever. Because if someone has access to them, he can produce fake proofs. 

It is worth noting that Algorithm \ref{Setup} is heavy in terms of computation load. However, this phase is done for function $F$ only once, and is not function of the values of the wires, inputs, or outputs.  This means that the cost of this algorithm amortizes over many proof generations about $F$. Thus, in this paper we do not deal with the setup phase. To see how setup phase calculations can be done in a multiparty protocol, refer to \cite {bowe2018multi}. 

\begin{algorithm}[!htbp]
\setstretch{1.5}
	\caption{Setup Phase Algorithm}
	\label{Setup}
	\begin{algorithmic}[1]
		\Statex
		\textbf{Input:}
		function $F$.
		\State
		Convert $F$ into an arithmetic circuit.
		\State
		Build QAP $\mathcal{Q} = \{T(x)$, $\left \{L_{i}(x)\right \}$, $\left \{R_{i}(x)\right \}$, $\left \{O_{i}(x)\right \} \}$ where $i \in \{0,\dots,m\}$.
		\State
		Choose parameters $s, \alpha,\beta,\gamma,\delta$ uniformly at random from $\mathbb{F}$.
		\State
		Evaluate polynomials $\left \{L_{i}(x)\right \}$, $\left \{R_{i}(x)\right \}$, $\left \{O_{i}(x)\right \}$ at the point $s$.
		\State 
		\hskip1em  Let $\mathbf{k}^{vk} =  \left [k^{vk}_i  \right ]_{i \in \mathcal{I}_{io}} = \left [\frac{\beta L_i(s) + \alpha R_i(s) + O_i(s)}{\gamma}\right ]_{i \in \mathcal{I}_{io}}$ where $\mathcal{I}_{io}$ is the set of indices of the wires that are the public input or the output of the arithmetic circuit.
		\State \hskip1em 
		Let $\mathbf{k}^{pk} =  \left [k^{pk}_i  \right ]_{i \in \mathcal{I}_{mid}} = \left [\frac{\beta L_i(s) + \alpha R_i(s) + O_i(s)}{\delta} \right ]_{i \in \mathcal{I}_{mid}}$ where $\mathcal{I}_{mid}$ is the set of indices of the wires that are not the public input nor the output of the arithmetic circuit.
		\State \hskip1em Let $\mathbf{t} = \left [t_j \right ]_{j=0}^{n-2} = \left [\frac{s^j T(s)}{\delta} \right ]_{j=0}^{n-2}$.
		\State
		Calculate $\mathcal{EK}$ and $\mathcal{VK}$ as below (cryptographic operations):
		\State \hskip1em 
		$\mathcal{EK} = \{ \llbracket \alpha \rrbracket_1$, $\llbracket \beta \rrbracket_1$, $\llbracket \beta \rrbracket_2$, $\llbracket \delta \rrbracket_1$, $\llbracket \delta \rrbracket_2$, $\left [ \llbracket L_{i}(s) \rrbracket_1 \right]_{i=0}^{m}$, $\left [ \llbracket R_{i}(s) \rrbracket_1 \right]_{i=0}^{m}$, $\left [ \llbracket L_{i}(s) \rrbracket_2 \right]_{i=0}^{m}$, $\llbracket \mathbf{k}^{pk} \rrbracket_1$, $\llbracket \mathbf{t} \rrbracket_1\}$.
	           \State \hskip1em 
		$\mathcal{VK} = \{e(\llbracket \alpha \rrbracket_1,\llbracket \beta \rrbracket_2)$, $\llbracket \gamma \rrbracket_2$, $\llbracket \delta \rrbracket_2$, $\llbracket \mathbf{k}^{vk} \rrbracket_1\}$.
	           \State Erase all values generated during the algorithm except $\mathcal{EK}$ and $\mathcal{VK}$.
		\Statex
		\textbf{Outputs}: $\mathcal{EK}$, $\mathcal{VK}$.
	\end{algorithmic}
\end{algorithm}

\item \textbf{Prover algorithm}: Prover algorithm is presented in Algorithm \ref{Prover}. This algorithm contains three main parts. In the first part the prover builds the arithmetic circuit and calculates the values of all wires. Remember that the multiplication gate labeled by $\omega^j$ multiplies $L(\omega^j)$ and $R(\omega^j)$, and outputs $O(\omega^j)$. 

In the second part, the prover runs the function $\mathsf{polynomial-division}$ to calculate coefficients of the polynomial $H(x) = \frac{P(x)} {T(x)}$. In the last part the prover runs the function $\mathsf{compute-proof}$ in order to generate the proof.\par

Function $\mathsf{polynomial-division}$, which is to calculate coefficients of the polynomial $H(x) = \frac{P(x)} {T(x)}$, needs some explanation.  Recall that $\deg H(x) = \frac{P(x)} {T(x)}< n-1$. If  the prover has the values of $P(x)$ and $T(x)$ in some $n$ distinct points $\mathcal{D}$, $|\mathcal{D}| = n$, he can calculate $H(x)$ in $\mathcal{D}$ by simply dividing $P(x)$ by $T(x)$ for those points. Thus, he can recover the coefficients of $H(x)$ by some interpolation.  In this algorithm, we  set $\mathcal{D} = \{\eta, \eta \omega, \dots, \eta \omega^{n-1} \}$ where $\eta \in \mathbb{F} \setminus \mathcal{S}$. We will explain the reason for this choice $\mathcal{D}$ later.  Recall that, the prover does not even have the coefficients $P(x)=L(x)R(x)-O(x)$. Instead, he has the values $L(x)$, $R(x)$, $O(x)$ in $\mathcal{S}$. Thus, we take the following steps:
\begin{itemize}
\item The coefficients  of $L(x)$, $R(x)$ and $O(x)$ are calculated  by interpolation over the values of $L(x)$, $R(x)$ and $O(x)$ in $\mathcal{S}$. This is done  efficiently by taking the $\mathscr{FFT}_{\mathcal{S}}^{-1}$ of the vectors $\mathbf{a}$, $\mathbf{b}$ and $\mathbf{c}$, containing  the values of $L(x)$, $R(x)$ and $O(x)$ in $\mathcal{S}$ respectively (see Lines 4-9 of Algorithm~\ref{Prover}).

\item  The values of  $L(x)$, $R(x)$ and $O(x)$ on the set $\mathcal{D}$ is obtained by take $\mathscr{FFT}_{\mathcal{D}}$ of their coefficients (see Lines 10-12 of Algorithm~\ref{Prover}).

\item $T(x) = x^n - 1$ is calculated at $n$ points of $\mathcal{D}$. 

\item For each $x \in \mathcal{D}$, calculate $H(x) = \frac{L(x) . R(x) - O(x)}{T(x)}$ (see Line 13 of Algorithm~\ref{Prover}).
\item Take $\mathscr{FFT}_{\mathcal{D}}^{-1}$ of the values of $H(x)$ on the set $\mathcal{D}$ to obtain the coefficients of $H(x)$.
(see Line 14 of Algorithm~\ref{Prover}).
\end{itemize}

We note that by definition, $T(x)$ and $P(x)$ both are zero on the set $\mathcal{S}$. Therefore, we need $\mathcal{D} \cap \mathcal{S} = \varnothing$, otherwise calculating  $\frac{P(x)}{T(x)}$ in $\mathcal{D}$ becomes undefined.  Choosing $\mathcal{D} = \{\eta, \eta \omega, \dots, \eta \omega^{n-1} \}$ for some $\eta \in \mathbb{F} \setminus \mathcal{S}$ has the following advantages: 
\begin{itemize}
\item $\mathcal{S} \cap \mathcal{D} = \varnothing$. The reason is that if $\mathcal{S} \cap \mathcal{D} \neq \varnothing$ there is at least a member $a$, where $a \in \mathcal{S}$ and $a \in \mathcal{D}$. Therefore $\exists i:  a = \omega^i$ and $\exists j:  a = \eta \omega^j$ so $\omega^i = \eta \omega^j$, that implies $\eta = \omega^{i-j}$. It means that $\eta \in \mathcal{S}$, a contradiction. 
\item Recall that $T(x) = x^n - 1$. So it would take less than $2n$ operations to compute $T(x)$ at $n$ points $\mathcal{D}$, because given $T(\eta \omega^j) = \eta^{j} - 1$, $T( \eta \omega^{(j+1)} ) = \eta^{(j+1)} - 1$ can be computed by one multiplication and one addition.
\item For each $x \in \mathcal{D}$ we have $T(x) \neq 0$. The reason is that $T(x)$ is of the degree $n$ so it can have at most $n$ distinct roots. On the other hand, we know that $T(x) = 0$ on the set $\mathcal{S}$ which has nothing in common with $\mathcal{D}$.
\item Computing fast Fourier transform on $\mathcal{S}$ and $\mathcal{D}$ is easy. See  Subsection \ref{FFT_learning}.
\end{itemize}

Now we focus on the computational cost of each step. In Algorithm \ref{Prover}, Lines 1 and 2 incur computation of the order $O(n)$ where $n$ is the number of multiplication gates. Lines 18 - 21 incur $O(m\kappa)$ operations, where $m$ is the number of index wires, and $\kappa$ is the security parameter. If the security parameter is too large, these lines can be dominant in terms of computational cost.
Recall that it needs at most $2n$ multiplications to calculate $T(x)$ on set $\mathcal{D}$, so Line 13 needs $O(n)$ operations. Lines 7 - 12 and 14 are usually the bulk of computation, and incur computation cost of the order $O(n \, \log(n))$.  Other lines of Algorithm \ref{Prover} incur a small amount of computation cost. So the computation cost of the prover is of the order $O(n \, \log(n))$. 

In this paper we propose a scheme by which the prover can delegate his task to $N = K+T$ semi-honest servers, where at most $T$ of them may collude.  Each machine will have computation cost of the order $O(\frac{n}{K}\, \log(\frac{n}{K}))$ and the prover's computation cost will be of the order $O(n)$.

\begin{algorithm}[!htbp]
\setstretch{1.5}
	\caption{Prover algorithm}
	\label{Prover}
	\begin{algorithmic}[1]
		\Statex
		\textbf{Inputs:}
		$\mathcal{EK}$, function $F$ and its inputs.
		\State
		Convert $F$ into arithmetic circuit and build QAP just like the setup phase.
		\State
		Compute the values of all wires in the arithmetic circuit.
		\State \textbf{function} $\mathsf{polynomial-division}$ ($P(x)$, $T(x)$)
		\State \hskip1em Let $\mathbf{a} = \left [a_j \right ]_{j = 0}^{n-1}$, where $a_j = L(\omega^j)$. 
		\State \hskip1em Let $\mathbf{b} = \left [b_j \right ]_{j = 0}^{n-1}$, where $b_j = R(\omega^j)$. 
		\State \hskip1em Let $\mathbf{c} = \left [c_j \right ]_{j = 0}^{n-1}$, where $c_j = O(\omega^j)$.
		\State \hskip1em Calculate $\mathbf{{a}'} = \mathscr{FFT}_{\mathcal{S}}^{-1}(\mathbf{a})$. Recall that $\mathcal{S} = \left \{1,\omega,\dots,\omega^{n-1} \right \}$.
		\State \hskip1em Calculate $\mathbf{{b}'} = \mathscr{FFT}_{\mathcal{S}}^{-1}(\mathbf{b})$.
		\State \hskip1em Calculate $\mathbf{{c}'} = \mathscr{FFT}_{\mathcal{S}}^{-1}(\mathbf{c})$.
		\State \hskip1em Calculate $\mathbf{{a}''} = \mathscr{FFT}_{\mathcal{D}}(\mathbf{{a}'})$, where $\mathcal{D} = \left \{\eta, \eta \omega, \dots, \eta \omega^{n-1} \right \}$.
		\State \hskip1em Calculate $\mathbf{{b}''} = \mathscr{FFT}_{\mathcal{D}}(\mathbf{{b}'})$.
		\State \hskip1em Calculate $\mathbf{{c}''} = \mathscr{FFT}_{\mathcal{D}}(\mathbf{{c}'})$.
		\State \hskip1em Calculate $\left [h_j \right ]_{j=0}^{n-1}$, where $h_j = ({a_j}''.{b_j}''-{c_j}'') \div T(\eta \omega^{j})$.
		\State \hskip1em Calculate $\left [{h}'_j \right]_{j=0}^{n-1} = \mathscr{FFT}_{\mathcal{D}}^{-1}(\left [h_j \right ]_{j=0}^{n-1})$.
		\State \textbf{return} $\left [{h}'_j \right]_{j=0}^{n-1}$
		\State \textbf{function} $\mathsf{compute-proof}$
		\State \hskip1em Choose secret parameters $r, q$ independently and uniformly at random from $\mathbb{F}$.
		\State \hskip1em Calculate $\llbracket L_r \rrbracket_1 = \llbracket \alpha \rrbracket_1 + \sum_{i = 0}^{m} W_i \llbracket L_i(s) \rrbracket_1 + r \llbracket \delta \rrbracket_1$ where $W_i$ is the value carried by wire $i$.
		\State \hskip1em Calculate $ \llbracket R_q \rrbracket _1 = \llbracket \beta \rrbracket_1 + \sum_{i = 0}^{m} W_i  \llbracket R_i(s) \rrbracket_1 + q \llbracket \delta \rrbracket_1$
		\State \hskip1em Calculate $ \llbracket R_q \rrbracket_2 = \llbracket \beta \rrbracket_2 + \sum_{i = 0}^{m} W_i \llbracket R_i(s)]_2 + q \llbracket \delta \rrbracket_2$
		\State \hskip1em Calculate $ \llbracket K_{r,q} \rrbracket_1 = q \llbracket L_r \rrbracket_1 + r \llbracket R_q \rrbracket_1 - rq \llbracket \delta \rrbracket_1 + \sum_{i \in \mathcal{I}_{mid}} W_i \llbracket k_i^{pk} \rrbracket_1 + \sum_{j=0}^{n-2} h'_j \llbracket t_j \rrbracket_1$
		\State \textbf{return} $\pi = \{ \llbracket L_r \rrbracket_1, \llbracket R_q \rrbracket_2, \llbracket K_{r,q} \rrbracket_1 \}$
		\Statex
		\textbf{Outputs}: $\pi$, public inputs, and public outputs of $F$.
	\end{algorithmic}
\end{algorithm}

\item \textbf{Verifier algorithm}: Verifier algorithm is presented in Algorithm \ref{Verifier}. As you can see, computation cost of Line 1 in this algorithm is proportional to $| \mathcal{I}_{io} |$ which is the number of public wires. Line 2 incurs constant computation cost. 

\begin{algorithm}[!htbp]
\setstretch{1.3}
	\caption{Verifier algorithm}
	\label{Verifier}
	\begin{algorithmic}[1]
		\Statex
		\textbf{Inputs}: $\mathcal{VK}$, $\pi$, public inputs, and public outputs of $F$.
\State Compute $I = \sum_{i \in \mathcal{I}_{io}} W_i \llbracket k_i^{vk}(s) \rrbracket_1$.
		\State Check $e \left( \llbracket L_r \rrbracket_1,  \llbracket R_q \rrbracket_2 \right) = e \left( \llbracket \alpha \rrbracket_{1}, \llbracket \beta \rrbracket_{2} \right) + e \left(I,  \llbracket \gamma \rrbracket_2 \right) + e \left( \llbracket K_{r,q} \rrbracket_1, \llbracket \delta \rrbracket_2 \right)$.
		\Statex
		\textbf{Output}: yes or no.
	\end{algorithmic}
\end{algorithm}

\end{itemize}

\section{preliminaries}
\label{Preliminaries2}
\subsection{Fast Fourier transform ($\mathscr{FFT}$) over finite field}
\label{FFT_learning}
Definitions of this subsection is taken from \cite{learnFFT}.
\begin{definition}
$\omega$ is a primitive $n$th root of unity in a computation structure (e.g. finite field), if $\omega^{n}=1$ but for no $m$ such that $0 < m < n$ is $\omega^{m}=1$.
\end{definition}

\begin{definition}
Let $\omega \in \mathbb{F}$ be a primitive $n$th root of unity. Fourier transform of an $n$ dimensional vector $\mathbf{a} \triangleq [a_0, a_2, \ldots, a_{n-1}]^\top$ over $\mathcal{S} = \left \{1, \omega, \omega^2, \dots, \omega^{n-1} \right \}$ is equal to $[f(1), f(\omega), \dots, f(\omega^{(n-1)})]^\top$, where $f(x) = a_0 + \dots + a_{n-1} x^{n-1}$. 

The specific structure of the Fourier transform allows us to develop it with the complexity $O(n \, \log(n))$, known as \emph{fast Fourier transform} denoted by $\mathscr{FFT_{\mathcal{S}}}(\mathbf{a})$. 
\end{definition}

The Fourier transform can be shown in matrix form as,

\begin{align}\label{FFT_def1}
\mathscr{FFT_{\mathcal{S}}}(\mathbf{a}) = \mathbf{F}_n \mathbf{a},
\end{align}
where
\begin{align}
\mathbf{F}_n \triangleq
	\begin{bmatrix}
	 1 & 1 & 1 & \dots & 1 \\ 
	 1 & \omega & \omega^{2} & \dots & \omega^{n-1}\\ 
	 \dots & \dots & \dots & \dots & \dots \\ 
	 1 & \omega^{n-1} & \omega^{2(n-1)} & \dots & \omega^{(n-1)^{2}}
	\end{bmatrix}.
\end{align}

The inverse of Fourier transform, denoted by $\mathscr{FFT}_{\mathcal{S}}^{-1}$,  is equal to, 
\begin{align}\label{FFT_inverse_def}
\mathscr{FFT}_{\mathcal{S}}^{-1}(
	\begin{bmatrix}
	f(1)
	\\ f(\omega)
	\\ \dots
	\\ f(\omega^{n-1})
	\end{bmatrix}) = \frac{1}{n} \mathbf{G}_n 
	\begin{bmatrix}
	f(1)
	\\ f(\omega)
	\\ \dots
	\\ f(\omega^{n-1})
	\end{bmatrix},
\end{align}
where $\mathbf{G}_n = \mathbf{F}_n^{-1}$ is equal to 
\begin{align}
 \mathbf{G}_n \triangleq
	\begin{bmatrix}
	 1 & 1 & 1 & \dots & 1 \\ 
	 1 & \omega^{-1} & \omega^{-2} & \dots & \omega^{-(n-1)}\\ 
	 \dots & \dots & \dots & \dots & \dots\\ 
	 1 & \omega^{-(n-1)} & \omega^{-2(n-1)} & \dots & \omega^{-(n-1)^{2}}
	\end{bmatrix}.
\end{align}

Equation \eqref{FFT_inverse_def} can also be represented by,
\begin{align}
\label{definition_element2}
a_i = \frac{1}{n} \sum_{\mu=0}^{n-1} \left (\omega^{-i} \right )^\mu f(\omega^\mu) \: \: \textup{for} \: \:
i\in \left \{0,1,\dots, n-1\right \}.
\end{align}

Let us define $\mathbf{a_e} = \left [a_i \right ]_{i \, = \, even}$ and $\mathbf{a_o} = \left [a_i \right ]_{i \, = \, odd}$, then Fourier transform of $\mathbf{a}$ can be written as:
\begin{align}\label{FFT_odd_even}
\mathscr{FFT}_{\mathcal{S}}(\mathbf{a}) = 
\begin{bmatrix}
\mathscr{FFT}_{\mathcal{S'}} (\mathbf{a_e}) \\ 
\mathscr{FFT}_{\mathcal{S'}} (\mathbf{a_e})
\end{bmatrix}
+ 
\begin{bmatrix}
1\\ 
\omega \\ 
\omega^2 \\ 
\dots \\ 
\omega^{n-1}
\end{bmatrix}
\circ
\begin{bmatrix}
\mathscr{FFT}_{\mathcal{S'}} (\mathbf{a_o}) \\ 
\mathscr{FFT}_{\mathcal{S'}} (\mathbf{a_o})
\end{bmatrix}
\end{align}
where $\mathcal{S'} = \left \{1, \omega^2, \omega^4, \dots, \left (\omega^2 \right )^{\frac{n}{2}-1} \right \}$ and the symbol $\circ$ denotes element-wise multiplication. This recursive structure has been used to develop algorithms that can do Fourier transform with complexity $O(n \, \log (n))$.

Calculating $\mathscr{FFT_{\mathcal{D}}}(\mathbf{a})$, where $\mathcal{D} = \left \{ \eta \omega^i \right \}_{i=0}^{n-1}$, $\eta \in \mathbb{F} \setminus \mathcal{S}$ incurs $O(n \, \log(n))$ operations too. Because if $\mathbf{\Xi}$ be a diagonal matrix, whose $i$th diagonal entry is $\eta^{i-1}$, we have the followings,
\begin{align}
\label{D_set1}
\mathscr{FFT_{\mathcal{D}}}(\mathbf{a}) = \mathscr{FFT_{\mathcal{S}}} \left (\mathbf{\Xi} \mathbf{a} \right) 
\end{align}
\begin{align}
\label{D_set2}
\mathscr{FFT}_{\mathcal{D}}^{-1}(\mathbf{a}) = \mathbf{\Xi}^{-1}  \mathscr{FFT}_{\mathcal{S}}^{-1} \left (\mathbf{a} \right )
\end{align}
Calculating the elements of ${\mathbf{\Xi}}$ or ${\mathbf{\Xi}^{-1}}$ requires $n$ multiplications, so the computation complexity of $\mathscr{FFT_{\mathcal{D}}}$ and $\mathscr{FFT}_{\mathcal{D}}^{-1}$ is of the order $O(n\, \log(n))$.

\subsection{Lagrange Sharing~\cite{yu2018lagrange}}
\label{LCC}
Consider a system including a  master and a cluster of $N \in \mathbb{N}$ servers. The master aims to share the set of private vectors $\{ \mathbf{x}_1, \dots,\mathbf{x}_K \}$, $K\in \mathbb{N}$,  for some finite field $\mathbb{F}$,  with those $N$ servers. The sharing must be such that if  any subset of $T$ servers collude, they  gain no information about the input data. Various approaches, such as ramp sharing \cite{bai2006strong} and Lagrange sharing \cite{yu2018lagrange} have been proposed for such sharing. In this paper, we use Lagrange Sharing \cite{yu2018lagrange}, which works as follows: 

Let $\{ \beta_{1},\dots,\beta_{K+T} \}$ and $\{ \alpha_{1},\dots, \alpha_{N}\}$ be two sets of some publicly known distinct non-zero points in the finite field $\mathbb{F}$ such that $\{\alpha_{\theta}\}_{\theta=1}^{N} \cap \{\beta_{j}\}_{j = 1}^{K} = \varnothing$. 

To code the secret inputs $\mathbf{x}_j \in \mathbb{F}^{M}$ for $j \in \left \{1,2,\dots,K \right \}$, first the master chooses $T$ vectors $\mathbf{y}_{j}$ for $j \in \left \{K+1,K+2,\dots,K+T \right \}$, 
 independently and uniformly at random from $\mathbb{F}^{M}$.  Then it forms the Lagrange coding polynomial $\mathbf{u}: \mathbb{F} \rightarrow \mathbb{F}^{M}$, 
defined as,  
\begin{align}\label{LCC_u}
\mathbf{u}(z) \triangleq 
 \sum_{j = 1}^{K} \mathbf{x}_{j} \prod_{k = 1, k \neq j}^{K+T} \frac{z - \beta_{k}}{\beta_{j} - \beta_{k}} + \sum_{j = K+1}^{K+T} \mathbf{y}_{j}  \prod_{k = 1, k \neq j}^{K+T} \frac{z - \beta_{k}}{\beta_{j} - \beta_{k}}.
\end{align}
 Finally in this stage, the master sends $\mathbf{u}(\alpha_{\theta})$ to Server $\theta$.

\section{The Proposed Scheme}
\label{sec:Proposed-Scheme}
Consider a system including a prover, $N = K + T$ semi-honest servers, and a globally known function $y = F(u,v)$ that the prover wants to generate a zkSNARK proof about it (see Section \ref{sec:Background}). Suppose that the function $F$ has a large arithmetic circuit and therefore it is a difficult task to produce proof $\pi$ about $F$. Therefore the prover may not be able to do the task alone, and he needs to delegate his task to the servers.  By semi-honest, we mean that the servers follow the algorithms correctly, but a subset of up to $T$ of them may collude to gain information about secret data. Note that if some of the servers are adversarial meaning that they don't follow the algorithm, the generated proof $\pi$ will fail the verification. Thus the prover itself can detect it, using the inherent zkSNARK verification ability.  In the other words, if for any reason, the generated proof is not valid, the prover will find out by running the verifier algorithm for his own.

As mentioned in Section \ref{sec:introduction}, the advantage of Trinocchio \cite{schoenmakers2016trinocchio} algorithm is that the prover can delegate his task to several untrusted servers. On the downside, in Trinocchio, the computation complexity of each server is equal to the main task, which is assumed to be large. On the other hand, the advantage of DIZK \cite{wu2018dizk} algorithm is that it partitions the task of the prover and gives a part to each server. The main  disadvantage of DIZK is that the servers must be trusted. In this section, we design an algorithm that has the advantage of both DIZK and Trinocchio. 

Recall that the computation complexity of zkSNARK $\mathsf{Prover  \ Algorithm}$ is equal to $O(n\log(n))$, which is driven by computing  $\mathbf{\hat{\hat{a}}} = \mathscr{FFT}_{\mathcal{D}} \left (\mathbf{\mathscr{FFT}_{\mathcal{S}}^{-1} (\mathbf{a})} \right )$, $\mathbf{\hat{\hat{b}}} = \mathscr{FFT}_{\mathcal{D}} \left (\mathbf{\mathscr{FFT}_{\mathcal{S}}^{-1} (\mathbf{b})} \right )$, $\mathbf{\hat{\hat{c}}} = \mathscr{FFT}_{\mathcal{D}} \left (\mathbf{\mathscr{FFT}_{\mathcal{S}}^{-1} (\mathbf{c})} \right )$,  in Lines 7-12 and also $\left [{h}'_j \right]_{j=0}^{n-1} = \mathscr{FFT}_{\mathcal{D}}^{-1}(\left [h_j \right ]_{j=0}^{n-1})$ in Line 14. We note that $n$ is can be easily in the order of $10^{10}$ or $2^{30}$. Thus, the factor $\log n$ is equal to $30$ which is considerable. If we can replace this factor by constant number, say 5, we can reduce the execution time six times which is very important.

In Subsection \ref{Proposed_Algorithm2}, we propose multiparty algorithms for computing $\mathscr{FFT}_{\mathcal{D}} \left (\mathbf{\mathscr{FFT}_{\mathcal{S}}^{-1} (\mathbf{a})} \right )$, and in Subsection \ref{H_algorithm}, we propose the main algorithm.

\subsection{The multiparty algorithm for computing $\mathscr{FFT}_{\mathcal{D}} \left (\mathbf{\mathscr{FFT}_{\mathcal{S}}^{-1} (\mathbf{a})} \right )$} 
\label{Proposed_Algorithm2}
Suppose that the prover has a large secret vector $\mathbf{a}$ of dimension $n$, and aims to compute $\mathscr{FFT}_{\mathcal{D}} \left (\mathbf{\mathscr{FFT}_{\mathcal{S}}^{-1} (\mathbf{a})} \right )$, using the cluster of $N$ semi-honest servers, where up to $T$ may collude.  
One approach would be consider $\mathscr{FFT}_{\mathcal{D}} \left (\mathbf{\mathscr{FFT}_{\mathcal{S}}^{-1} (\mathbf{a})} \right )$ as a matrix multiplication problem $\frac{1}{n}\mathbf{G}_{\mathcal{D}} \mathbf{F}_{\mathcal{S}}  \mathbf{a}$, as defined in Subsection \ref{FFT_learning}. Then, we can use secure multiparty computation for massive data or secure matrix multiplication methods \cite{nodehi2021secure, najarkolaei2020coded} by partitioning and securely sharing vector $\mathbf{a}$ and partitioning and sharing matrix $\mathbf{F}_{\mathcal{S}}$ $\mathbf{G}_{\mathcal{D}}$ with the servers. Then each server simply multiply what it received and sends it back the  servers. However, using this approach, we lose the Fourier transform structure, and the computation complexity would be of the order $O({\frac{n^{2}}{K}})$ for each server.

Here we propose an alternative approach in Algorithm \ref{fft_MPC_v2} such that the complexity of computation in each server is equal to $O(\frac{n}{K} \log(\frac{n}{K}))$, and the complexity of computation in the prover is equal to $O(Kn)$. 

\begin{algorithm}[!htbp]
\setstretch{1.75}
\caption{Multiparty algorithm for computing $\mathbf{\hat{\hat{a}}} = \mathscr{FFT}_{\mathcal{D}} \left (\mathbf{\mathscr{FFT}_{\mathcal{S}}^{-1} (\mathbf{a})} \right )$}
\label{fft_MPC_v2}
\begin{algorithmic}[1]
\Statex
\textbf{Input}: vector $\mathbf{a}=\left [a_i \right ]_{i=0}^{n-1}$ of the length $n$.
\Statex
The prover does the following  steps:
\State \hskip1em 
partitions $\mathbf{a}$ into $K$ vectors $\mathbf{a}^{(j)}=\left [ a_{Kt+j-1}  \right ]_{t=0}^{\frac{n}{K}-1}$ for $j \in \{1,\dots,K\}$.
\State \hskip1em 
picks $T$ vectors $\mathbf{v}_j \in \mathbb{F}^{\frac{n}{K}}$, $j \in \{K+1,\dots,K+T\}$, independently and uniformly at random.
\State \hskip1em 
defines $\mathbf{u}(z)$ as  (\ref{FFT_first_algorthm_u}) and sends $\mathbf{u}(\alpha_\theta)$ to Server $\theta$, for  $\theta \in \{1,\dots,N\}$.
\Statex
Server $\theta$ does the following  steps:
\State \hskip1em
computes $\mathscr{FFT}_{\mathcal{S'}}^{-1}(\mathbf{u}(\alpha_\theta))$ for $\mathcal{S'} = \left \{1, \omega^K, \dots, \omega^{K \left (\frac{n}{K}-1 \right)} \right \}$.
\State \hskip1em
picks $T$ vectors $\mathbf{v}_{j}^{(\theta)} \in \mathbb{F}^{\frac{n}{K}}$, $j \in \{K+1,\dots,K+T\}$, independently and uniformly at random.

\State \hskip1em 
defines  $\mathbf{u}^{(\theta)}(z)$ as  (\ref{FFT_version2_algorthm_u}), and sends $\mathbf{u}^{(\theta)}(\alpha_\gamma)$ to Server $\gamma$, for  $\gamma \in \{1,\dots,N\}$.
\State \hskip 1em 
calculates $\mathbf{u}_{\theta} \triangleq \sum_{\gamma=1}^{N} \mathbf{u}^{(\gamma)}(\alpha_\theta)$, upon receiving $\mathbf{u}^{(\gamma)}(\alpha_\theta)$, $\gamma \in \{1, \ldots, N\}$. 
\State \hskip1em
computes $\mathscr{FFT}_{\mathcal{D'}}(\mathbf{u}_\theta)$ where $\mathcal{D'} = \left \{\eta^K, \eta^K \omega^K, \dots, \eta^K \omega^{K \left (\frac{n}{K}-1 \right)} \right \}$, and sends it to the prover.
\Statex
Upon receiving $\mathscr{FFT}_{\mathcal{D'}}(\mathbf{u}_\theta)$, $\theta \in \{1,\dots,N\}$,   the prover does the following steps:
\State \hskip 1em 
forms $\mathbf{u}_{p}(z)$ as defined in (\ref{prover_decoding_poly}), and evaluates $\mathbf{u}_{p}(\beta_j)$ for $j \in \{1,\dots,K \}$.
\State \hskip1em
calculates the vector $\mathbf{\hat{\hat{a}}} =\left [\hat{\hat{a}}_i \right ]_{i=0}^{n-1}$ as $\hat{\hat{a}}_i \triangleq \sum_{j=1}^{K} \left (\omega^{i} \eta \right )^{(j-1)} \left( \mathbf{u}_{p}(\beta_j) \right )_{i \mod \frac{n}{K}}$.
\Statex
\textbf{Output}: $\mathbf{\hat{\hat{a}}}$
\end{algorithmic}
\end{algorithm}


In step 1,  vector $\mathbf{a}$ of length $n$ is partitioned into $K$ vectors of length $\frac{n}{K}$. We have assumed that $n$ and $K$ are powers of 2, so $n$ is divisible by $K$. 

In step 2, the prover picks $T$ vectors $\mathbf{v}_j \in \mathbb{F}^{\frac{n}{K}}$ independently and uniformly at random. These random vectors are used to guarantee  the privacy of $\mathbf{a}$ against $T$ colluding servers.

In step 3, to encode the information, the prover evaluates the Lagrange polynomial $\mathbf{u}(z)$ according to (\ref{FFT_first_algorthm_u}) in public points $\left \{\alpha_\theta \right \}_{\theta=1}^{N}$. It then sends $\mathbf{u}(\alpha_\theta)$ to Server $\theta$. For more information about Lagrange polynomial coding refer to Section \ref{LCC}. 

\begin{align}
\label{FFT_first_algorthm_u}
\mathbf{u}(z) =  \sum_{j = 1}^{K} \mathbf{a}^{(j)}  \prod_{k = 1, k \neq j}^{K+T} \frac{z - \beta_{k}}{\beta_{j} - \beta_{k}} + \sum_{j = K+1}^{K+T} \mathbf{v}_{j}  \prod_{k = 1, k \neq j}^{K+T} \frac{z - \beta_{k}}{\beta_{j} - \beta_{k}}.
\end{align}

In step 4, Server $\theta$ computes $\mathscr{FFT}_{\mathcal{S'}}^{-1}(\mathbf{u}(\alpha_\theta))$. 

In step 5, Server $\theta$ picks $T$ random vectors $\mathbf{v}_{j}^{(\theta)} \in \mathbb{F}^{\frac{n}{K}}$, for $j \in \{K+1,\dots,K+T\}$, independently and uniformly at random. 

Next, in step 6, Server $\theta$ forms $\mathbf{u}^{(\theta)}(z)$ as,  
\begin{align}\label{FFT_version2_algorthm_u}
& \mathbf{u}^{(\theta)}(z) = \sum_{j = 1}^{K} \mathbf{x}_{j}^{(\theta)}  \prod_{k = 1, k \neq j}^{K+T} \frac{z - \beta_{k}}{\beta_{j} - \beta_{k}} + \sum_{j = K+1}^{K+T} \mathbf{v}_{j}^{(\theta)}  \prod_{k = 1, k \neq j}^{K+T} \frac{z - \beta_{k}}{\beta_{j} - \beta_{k}},
\end{align}
where 
\begin{align}\label{definition_x}
x_{j,t}^{(\theta)} \triangleq \left ( \frac{K}{n} \sum_{l =1}^{K} \omega^{-(Kt + j - 1) (l-1)} 
 \mathscr{FFT}_{\mathcal{S'}}^{-1} \left (\mathbf{u} \left (\alpha_\theta \right) \right) \prod_{k = 1, k \neq \theta}^{N} \frac{\beta_l - \alpha_{k}}{\alpha_{\theta} - \alpha_{k}} \right )_{{Kt + j - 1 \mod \frac{n}{K}}}.
 \end{align}

In step 6, Server $\theta$ sends $\mathbf{u}^{(\theta)}(\alpha_\gamma)$ to Server $\gamma$.

In step 7, Server $\theta$ calculates $\mathbf{u}_\theta \triangleq \sum_{\gamma=1}^{N} \mathbf{u}^{(\gamma)}(\alpha_\theta)$, where $\mathbf{u}^{(\gamma)}(\alpha_\theta)$ is the data sent from Server $\gamma$ to Server $\theta$. 

In step 8, Server $\theta$ computes $\mathscr{FFT}_{\mathcal{D'}}(\mathbf{u}_\theta)$, where $\mathcal{D'} = \left \{\eta^K, \eta^K \omega^K, \dots, \eta^K \omega^{K \left (\frac{n}{K}-1 \right)} \right \}$, and sends it to the prover. 

In step 9, the prover forms the polynomial $\mathbf{u}_{p}(z)$ as,
\begin{align} \label{prover_decoding_poly}
\mathbf{u}_{p}(z) \triangleq \sum_{\theta = 1}^{N} \mathscr{FFT}_{\mathcal{D'}}(\mathbf{u}_\theta) \prod_{k = 1, k \neq \theta}^{N} \frac{z - \alpha_{k}}{\alpha_{\theta} - \alpha_{k}},
\end{align}
and evaluates  $\mathbf{u}_{p}(\beta_j)$ for $j \in \{1,\dots,K \}$.

In step 10, the prover calculates the vector $\mathbf{\hat{\hat{a}}} =\left [\hat{\hat{a}}_i \right ]_{i=0}^{n-1}$ as $\hat{\hat{a}}_i \triangleq \sum_{j=1}^{K} \left (\omega^{i} \eta \right )^{(j-1)} \left( \mathbf{u}_{p}(\beta_j) \right )_{i \mod \frac{n}{K}}$.


In what follows, we prove the correctness and privacy of the proposed scheme. 

\subsubsection{Correctness} In this part we prove that Algorithm \ref{fft_MPC_v2} is correct and gives the desired output.

\begin{lemma}\label{lemma1}
Let $\mathbf{u}(z)$ be defined as in (\ref{LCC_u}). The following equality always holds.
\begin{align}
\sum_{\theta=1}^{N} \mathbf{u}(\alpha_\theta)  \prod_{k = 1, k \neq \theta}^{N} \frac{\beta_{j} - \alpha_{k}}{\alpha_{\theta} - \alpha_{k}} = \mathbf{x}_{j}
\end{align}
\end{lemma}
\begin{proof}
Let define $\mathbf{u}'(z)$, a polynomial of degree $N-1$ as,
\begin{align}
\mathbf{u}'(z) \triangleq \sum_{\theta=1}^{N} \mathbf{u}(\alpha_\theta)  \prod_{k = 1, k \neq \theta}^{N} \frac{z - \alpha_{k}}{\alpha_{\theta} - \alpha_{k}}.
\end{align}
We note that $\mathbf{u}'(\alpha_{\theta}) = \mathbf{u}(\alpha_{\theta})$ for $\theta \in \{1,\dots,N \}$. Recall that $\mathbf{u}(z)$ is a polynomial of degree $K+T-1=N-1$. Because these two polynomials of the degree $N-1$ have the same values in $N$ different points, we can conclude $\mathbf{u}'(z) = \mathbf{u}(z)$ for any $z \in \mathbb{F}$. 

Also by considering the definition of $\mathbf{u}(z)$ in (\ref{LCC_u}), we can see $\mathbf{u}(\beta_j) = \mathbf{x}_j$ for $j \in \{1, \dots, K\}$. As a result, we have $\mathbf{u}'(\beta_j) = \mathbf{u}(\beta_j) = \mathbf{x}_j$. 
\end{proof}

\begin{lemma}\label{lemma2}
Let $\mathbf{\hat{a}} = \left [\hat{a}_i \right ]_{i=0}^{n-1} \triangleq \mathscr{FFT}_{\mathcal{S}}^{-1} \left ( \mathbf{a} \right )$, partitioned into $K$ vectors $\mathbf{\hat{a}}^{(j)} \triangleq \left [\hat{a}_{Kt+j-1} \right ]_{t = 0}^{\frac{n}{K}-1}$,  $j \in \{1,\dots,K \}$. Then we have 
$\sum_{\gamma = 1}^{N} \mathbf{x}_{j}^{(\gamma)} = \mathbf{\hat{a}}^{(j)}$, where 
$\sum_{\gamma = 1}^{N} \mathbf{x}_{j}^{(\gamma)} = \left [x_{j,t}^{(\gamma)} \right ]_{t = 0}^{\frac{n}{K}-1}$ is defined in (\ref{definition_x}). 
\end{lemma}
\begin{proof}
It is sufficient to prove $\sum_{\gamma = 1}^{N} x_{j,t}^{(\gamma)} = \left ( \mathscr{FFT}_{\mathcal{S}}^{-1} \left ( \mathbf{a} \right) \right )_{{Kt + j - 1 \mod \frac{n}{K}}} $ as,
\begin{align}
\begin{split}
\sum_{\gamma = 1}^{N} x_{j,t}^{(\gamma)} & \stackrel{(a)}{=} \sum_{\gamma = 1}^{N} \left ( \frac{K}{n} \sum_{l =1}^{K} \omega^{-(Kt + j - 1) (l-1)} \mathscr{FFT}_{\mathcal{S'}}^{-1} \left (\mathbf{u} \left (\alpha_\gamma \right) \right) 
\prod_{k = 1, k \neq \gamma}^{N} \frac{\beta_l - \alpha_{k}}{\alpha_{\gamma} - \alpha_{k}} \right )_{{Kt + j - 1 \mod \frac{n}{K}}}\\
& \stackrel{(b)}{=} \left ( \frac{K}{n} \sum_{l =1}^{K} \omega^{-(Kt + j - 1) (l-1)} 
\mathscr{FFT}_{\mathcal{S'}}^{-1} \left ( \sum_{\gamma = 1}^{N} \mathbf{u} \left (\alpha_\gamma \right) \prod_{k = 1, k \neq \gamma}^{N} \frac{\beta_l - \alpha_{k}}{\alpha_{\gamma} - \alpha_{k}} \right) \right )_{{Kt + j - 1 \mod \frac{n}{K}}} \\
& \stackrel{(c)}{=} \left ( \frac{K}{n} \sum_{l =1}^{K} \omega^{-(Kt + j - 1) (l-1)} 
\mathscr{FFT}_{\mathcal{S'}}^{-1} \left ( \mathbf{a}^{(l)} \right) \right )_{{Kt + j - 1 \mod \frac{n}{K}}} \\
& \stackrel{(d)}{=} \left ( \mathscr{FFT}_{\mathcal{S}}^{-1} \left ( \mathbf{a} \right) \right )_{{Kt + j - 1 \mod \frac{n}{K}}}, \\
\end{split}
\nonumber
\end{align}
where (a) follows from the definition of $x_{j,t}^{(\gamma)}$ in \eqref{definition_x}, (b) holds because of the linearity of Fourier transform, (c) follows from Lemma \ref{lemma1}, and (d) is based on the definition of Fourier transform.
\end{proof}

The following theorem establishes the correctness of Algorithm~\ref{fft_MPC_v2}.

\begin{theorem}\label{corr_proof}
The vector $\mathbf{\hat{\hat{a}}}$, calculated in step 10 of Algorithm \ref{fft_MPC_v2}, is equal to $\mathscr{FFT}_{\mathcal{D}} \left ( \mathscr{FFT}_{\mathcal{S}}^{-1} \left (\mathbf{a} \right) \right )$.
\end{theorem}
\begin{proof}
We have 
\begin{align}
\begin{split} 
\hat{\hat{a}}_i & \stackrel{(a)}{=} \sum_{j=1}^{K} \left (\omega^{i} \eta \right )^{(j-1)} \left( \mathbf{u}_{p}(\beta_j) \right )_{i \mod \frac{n}{K}}\\
& \stackrel{(b)}{=} \sum_{j=1}^{K} \left (\omega^{i} \eta \right )^{(j-1)} \left( 
\sum_{\theta = 1}^{N} \mathscr{FFT}_{\mathcal{D'}}(\mathbf{u}_\theta) \prod_{k = 1, k \neq \theta}^{N} \frac{\beta_{j} - \alpha_{k}}{\alpha_{\theta} - \alpha_{k}}
 \right )_{i \mod \frac{n}{K}}\\
& \stackrel{(c)}{=} \sum_{j=1}^{K} \left (\omega^{i} \eta \right )^{(j-1)} \left( \mathscr{FFT}_{\mathcal{D'}} \left (
\sum_{\theta = 1}^{N} \mathbf{u}_\theta \prod_{k = 1, k \neq \theta}^{N} \frac{\beta_{j} - \alpha_{k}}{\alpha_{\theta} - \alpha_{k}}
\right ) \right )_{i \mod \frac{n}{K}}\\
& \stackrel{(d)}{=} \sum_{j=1}^{K} \left (\omega^{i} \eta \right )^{(j-1)} \left( \mathscr{FFT}_{\mathcal{D'}} \left (
\sum_{\theta = 1}^{N} \sum_{\gamma=1}^{N} \mathbf{u}^{(\gamma)}(\alpha_\theta) \prod_{k = 1, k \neq \theta}^{N} \frac{\beta_{j} - \alpha_{k}}{\alpha_{\theta} - \alpha_{k}}
\right ) \right )_{i \mod \frac{n}{K}}\\
& \stackrel{}{=} \sum_{j=1}^{K} \left (\omega^{i} \eta \right )^{(j-1)} \left( \mathscr{FFT}_{\mathcal{D'}} \left (
\sum_{\gamma = 1}^{N} \sum_{\theta=1}^{N} \mathbf{u}^{(\gamma)}(\alpha_\theta)  \prod_{k = 1, k \neq \theta}^{N} \frac{\beta_{j} - \alpha_{k}}{\alpha_{\theta} - \alpha_{k}}
\right ) \right )_{i \mod \frac{n}{K}}\\
& \stackrel{(e)}{=} \sum_{j=1}^{K} \left (\omega^{i} \eta \right )^{(j-1)} \left( \mathscr{FFT}_{\mathcal{D'}} \left (
\sum_{\gamma = 1}^{N} \mathbf{x}_{j}^{(\gamma)}
\right ) \right )_{i \mod \frac{n}{K}}\\
& \stackrel{(f)}{=} \sum_{j=1}^{K} \left (\omega^{i} \eta \right )^{(j-1)} \left( \mathscr{FFT}_{\mathcal{D'}} \left (
\mathbf{\hat{a}}^{(j)}
\right ) \right )_{i \mod \frac{n}{K}}\\
& \stackrel{(g)}{=} \left ( \mathscr{FFT}_{\mathcal{D}} \left ( \mathbf{\hat{a}} \right ) \right )_{i} \\
& \stackrel{(h)}{=} \left ( \mathscr{FFT}_{\mathcal{D}} \left ( \mathscr{FFT}_{\mathcal{S}}^{-1} \left (\mathbf{a} \right) \right ) \right )_{i}, \\
\end{split}
\nonumber
\end{align}
where (a) is based on the definition of $\hat{\hat{\mathbf{a}}}$ in step 10 of Algorithm \ref{fft_MPC_v2}, (b) is based on the definition (\ref{prover_decoding_poly}) of $\mathbf{u}_{p}(z)$, (c) holds because of the linearity of Fourier transform, (d) follows from the definition of $\mathbf{u}_{\theta}$,  in step 7 of Algorithm \ref{fft_MPC_v2},  (e) follows from Lemma \ref{lemma1}, (f) follows from Lemma \ref{lemma2}, (g) follows from the definition of Fourier transform, and (h) follows from the definition of $\mathbf{\hat{a}}$.
\end{proof}

\subsubsection{Privacy}
To prove the privacy of Algorithm \ref{fft_MPC_v2}, we must show that if an arbitrary subset of at most $T$ servers collude, and share all the information they have, they cannot obtain any information about the vector $\mathbf{a}$. Without loss of generality, let us assume Servers $1,\dots,T$ collude. We note that Server $\theta$ receives $\mathbf{u}(\alpha_{\theta})$ and $\mathbf{u}^{(\gamma)}(\alpha_{\theta})$, for all $\gamma \in \{1,\dots,N\}$. As short hand notations, we define $\mathbf{U} \triangleq \left [\mathbf{u}(\alpha_{1}), \dots,\mathbf{u}(\alpha_{T}) \right ]$ where $\mathbf{u}(z)$ is defined in (\ref{FFT_first_algorthm_u}), and $\mathbf{\tilde{U}}$ as,
\begin{align}\label{priv_proof_al3_1_U}
& \mathbf{\tilde{U}} \triangleq 
\begin{bmatrix}
\mathbf{u}^{(1)}(\alpha_1) & \mathbf{u}^{(1)}({\alpha_2}) & \dots & \mathbf{u}^{(1)}({\alpha_T})\\ 
\mathbf{u}^{(2)}(\alpha_1) & \mathbf{u}^{(2)}({\alpha_2}) & \dots & \mathbf{u}^{(2)}({\alpha_T})\\ 
\dots & \dots & \dots & \dots\\ 
\mathbf{u}^{(N)}(\alpha_1) & \mathbf{u}^{(N)}(\alpha_2) & \dots & \mathbf{u}^{(N)}({\alpha_T})
\end{bmatrix}.
\end{align}

Thus to  prove  privacy, we need to show that $I \left (\mathbf{U}, \mathbf{\tilde{U}};\mathbf{a} \right) = 0$.

Let us define  matrices $\mathbf{A}$, $\mathbf{V}$, $\mathbf{P}$, $\mathbf{Q}$, and $\mathbf{\tilde{X}}$ as follows. 

Let $\mathbf{A} \triangleq \left [\mathbf{a}^{(1)}, \dots, \mathbf{a}^{(K)} \right]$, where $\mathbf{a}^{(1)}, \dots, \mathbf{a}^{(K)}$ are defined in step 1 of Algorithm \ref{fft_MPC_v2}, and $\mathbf{V} \triangleq \left [\mathbf{v}_{K+1}, \dots,\mathbf{v}_{K+T} \right]$ where $\mathbf{v}_{K+1}, \dots, \mathbf{v}_{K+T}$ are random vectors that are chosen randomly in step 2 of Algorithm \ref{fft_MPC_v2}. We also define constant $\mathbf{P} = \left [p_{i,j} \right]_{i=1,j =1}^{K,T}$ and $\mathbf{Q}= \left [q_{i,j} \right]_{i=K+1,j =1}^{K+T,T}$ as,
\begin{align}\label{priv_proof_3}
\begin{split}
& p_{i,j} \triangleq \prod_{k = 1, k \neq i}^{K+T} \frac{\alpha_j - \beta_{k}}{\beta_{i} - \beta_{k}} \; \; \textup{for} \; \;  i = 1, \dots, K , \: j = 1,\dots,T,\\
& q_{i,j} \triangleq \prod_{k = 1, k \neq i}^{K+T} \frac{\alpha_j - \beta_{k}}{\beta_{i} - \beta_{k}} \; \; \textup{for} \; \;  i = K+1, \dots,K+T , \: 
j = 1,\dots,T.
\end{split}
\end{align}

According to (\ref{FFT_version2_algorthm_u}), we can write $\mathbf{U}$ as, 
\begin{align}\label{priv_proof_2}
\mathbf{U} = \mathbf{A} \mathbf{P} + \mathbf{V} \mathbf{Q}. 
\end{align}

We note that matrix $\mathbf{Q}$ is a full rank matrix. This is because by multiplying the $i$th row of the matrix $\mathbf{Q}$ by a constant non-zero number $\prod_{k = 1, k \neq i}^{K+T} \beta_{i} - \beta_{k}$ and dividing the $j$th column of $\mathbf{Q}$ by constant non-zero number $\prod_{k = 1}^{K+T} \alpha_j - \beta_{k}$, for $i \in \{K+1, \dots, K+T \}, j \in \{1, \dots, T \}$, we will reach to a square Cauchy matrix, which is full rank. 

Considering (\ref{definition_x}), $\mathbf{\tilde{X}}$ is defined as,
\begin{align}\label{priv_proof_al3_1}
\mathbf{\tilde{X}} \triangleq
\begin{bmatrix}
\mathbf{x}_{1}^{(1)} &\mathbf{x}_{2}^{(1)} & \dots & \mathbf{x}_{K}^{(1)}\\ 
\mathbf{x}_{1}^{(2)} &\mathbf{x}_{2}^{(2)} & \dots & \mathbf{x}_{K}^{(2)}\\
\dots & \dots & \dots& \dots\\ 
\mathbf{x}_{1}^{(N)} & \mathbf{x}_{2}^{(N)} &\dots & \mathbf{x}_{K}^{(N)}
\end{bmatrix}.
\end{align}


\begin{lemma} \label{Markov99} $I \left (\mathbf{U} ; \mathbf{A} \right) = 0$.
\end{lemma}

\begin{proof}
we have
\begin{align}
\begin{split}
I \left (\mathbf{U};\mathbf{A} \right) & \stackrel{(a)}{=} H(\mathbf{U}) - H(\mathbf{U} | \mathbf{A})\\
& \stackrel{(b)}{=} H(\mathbf{U}) -  H(\mathbf{A} \mathbf{P} + \mathbf{V} \mathbf{Q} | \mathbf{A}) \\
& \stackrel{(c)}{=} H(\mathbf{U}) -  H(\mathbf{V} \mathbf{Q} | \mathbf{A})\\ 
& \stackrel{(d)}{=} H(\mathbf{U}) -  H(\mathbf{V} \mathbf{Q})\\
& \stackrel{(e)}{=} H(\mathbf{U}) -  \frac{n}{K} T \log (|\mathbb{F}|)\\
& \stackrel{(f)}{\leq} 0,
\end{split}
\nonumber
\end{align}
where (a) follows from the definition of the mutual information, (b) relies on (\ref{priv_proof_2}), (c) holds because matrix $\mathbf{P}$ is a constant matrix, (d) holds because $\mathbf{V}$ is chosen independently of matrix $\mathbf{A}$, (e) holds because  matrix $\mathbf{V}$ has uniform distribution over $\mathbb{F}^{\frac{n}{K}\times T}$, and matrix $\mathbf{Q}$ is full rank, and (f) holds because $\mathbf{U} \in \mathbb{F}^{\frac{n}{K}\times T}$, so we have $H(\mathbf{U}) \leq \frac{n}{K} \times T \times \log (|\mathbb{F}|)$.

Since  mutual information is always non-negative, thus $I \left (\mathbf{U};\mathbf{A} \right) \leq 0$ means   $I\left (\mathbf{U};\mathbf{A} \right) = 0$.
\end{proof}

\begin{lemma}\label{Markov100}
Given $\mathbf{U}$,   $\mathbf{A} \rightarrow \mathbf{\tilde{X}} \rightarrow \mathbf{\tilde{U}}$ is a Markov chain.
\end{lemma}
\begin{proof}
\begin{align}
\begin{split}
I \left (\mathbf{A};\mathbf{\tilde{U}}| \mathbf{\tilde{X}}, \mathbf{U} \right) & \stackrel{}{=} H \left (\mathbf{A} | \mathbf{\tilde{X}}, \mathbf{U} \right) - H \left (\mathbf{A} |\mathbf{\tilde{U}}, \mathbf{\tilde{X}}, \mathbf{U} \right)\\
& \stackrel{(a)}{=} - H \left (\mathbf{A} |\mathbf{\tilde{U}}, \mathbf{\tilde{X}}, \mathbf{U} \right), \\
\end{split}
\nonumber
\end{align}
where (a) is because if someone has $\mathbf{\tilde{X}}$ and globally known parameters $\omega, n, K, \{\alpha_\theta \}_{\theta = 1}^{N}, \{\beta_j \}_{j = 1}^{K}$, he can perform the roles of the servers and the prover in Algorithm \ref{fft_MPC_v2} from step 5 to the end and calculate $\mathbf{\hat{\hat{a}}}$. So he can calculate $ \mathbf{a} = \mathscr{FFT}_{\mathcal{S}} \left (\mathbf{\mathscr{FFT}_{\mathcal{D}}^{-1} (\mathbf{\hat{\hat{a}}})} \right )$. So $\mathbf{a}$ (or equivalently $\mathbf{A}$) can be recovered from $\mathbf{\tilde{X}}$. 

Finally, since the entropy function $H$ and mutual information $I$ are always non-negative, we conclude $I \left (\mathbf{A};\mathbf{U}| \mathbf{\tilde{X}}, \mathbf{U} \right) = 0$. 
\end{proof}

Theorem \ref{privacy_theorem} establishes the  privacy of algorithm. 

\begin{theorem} \label{privacy_theorem}
$I \left (\mathbf{U}, \mathbf{\tilde{U}};\mathbf{a} \right) = 0$.
\end{theorem}
\begin{proof}
We have 
\begin{align}
\begin{split}
I \left (\mathbf{U}, \mathbf{\tilde{U}};\mathbf{a} \right) & \stackrel{(a)}{=}  I \left (\mathbf{U}, \mathbf{\tilde{U}};\mathbf{A} \right)\\
& \stackrel{(b)}{=} I \left (\mathbf{U};\mathbf{A} \right) + I \left (\mathbf{\tilde{U}};\mathbf{A} | \mathbf{U} \right)\\
& \stackrel{(c)}{=} I \left (\mathbf{\tilde{U}};\mathbf{A} | \mathbf{U} \right)\\
& \stackrel{(d)}{\leq} I \left (\mathbf{\tilde{U}};\mathbf{\tilde{X}} | \mathbf{U} \right) \\
& \stackrel{(e)}{=} H(\mathbf{\tilde{U}} | \mathbf{U}) - H(\mathbf{\tilde{U}} | \mathbf{\tilde{X}}, \mathbf{U}) \\
& \stackrel{(f)}{=} H(\mathbf{\tilde{U}} | \mathbf{U}) - H(\mathbf{\tilde{X}} \mathbf{P} + \mathbf{\tilde{V}} \mathbf{Q} | \mathbf{\tilde{X}}, \mathbf{U})\\
& \stackrel{(g)}{=} H(\mathbf{\tilde{U}} | \mathbf{U}) - H(\mathbf{\tilde{V}} \mathbf{Q} | \mathbf{\tilde{X}}, \mathbf{U})\\
& \stackrel{(h)}{=} H(\mathbf{\tilde{U}} | \mathbf{U}) - H(\mathbf{\tilde{V}} \mathbf{Q}) \\
& \stackrel{(i)}{\leq} H(\mathbf{\tilde{U}}) - H(\mathbf{\tilde{V}} \mathbf{Q}) \\
& \stackrel{(j)}{=} H(\mathbf{\tilde{U}}) -  \frac{Nn}{K}T  \log (|\mathbb{F}|)\\
& \stackrel{(k)}{\leq} 0,
\end{split}
\nonumber
\end{align}
where (a) holds because $\mathbf{A}$ is made up of partitions of $\mathbf{a}$, (b) is based on definition of the mutual information, (c) relies on Lemma \ref{Markov99}, (d) is obtained from Lemma \ref{Markov100}, (e) is based on the definition of mutual information, (f) holds because of the equation $\mathbf{\tilde{U}} = \mathbf{\tilde{X}} \mathbf{P} + \mathbf{\tilde{V}} \mathbf{Q}$ which is the result of (\ref{FFT_version2_algorthm_u}), (g) holds because $\mathbf{P}$ is a constant known matrix,  (h) holds because $\mathbf{Q}$ is a known constant matrix, and matrix $\mathbf{\tilde{V}}$ is independent of $\mathbf{\tilde{X}}$ and $\mathbf{U}$, (i) holds because elimination of the condition doesn't decrease the entropy, (j) follows from $H(\mathbf{\tilde{V}} \mathbf{Q}) = \frac{Nn}{K}  T  \log (|\mathbb{F}|)$ which is because $\mathbf{Q}$ is a full rank matrix, and $\mathbf{\tilde{V}}$ has uniform distribution on $\mathbb{F}^{\frac{N.n}{K} \times T}$, and (k) is because  $\mathbf{\tilde{U}}  \in \mathbb{F}^{\frac{N.n}{K}\times T}$. 

Since  mutual information is non-negative, then, from above, we conclude  $I \left (\mathbf{U}, \mathbf{\tilde{U}};\mathbf{A} \right) = 0$. 
\end{proof}
\subsubsection{Computation Complexity}
In this part we explore the computation complexity of Algorithm \ref{fft_MPC_v2}. The computation complexity of each step in this algorithm is as follows:
\begin{itemize}
\item Step 3 has the cost of evaluating $\mathbf{u}(z)$, a polynomial of degree $N-1$, at $N$ points. This has the complexity of $O(\frac{n}{K}N \; \log(N))$ 
\item Step 4 has the cost of computing inverse Fourier transform of a vector of the length $\frac{n}{K}$ with the computation complexity of $O(\frac{n}{K} \; \log (\frac{n}{K}))$. 
\item Step 6 has the cost of evaluating $\mathbf{u}^{(\theta)}(z)$, a polynomial of degree $N-1$, at $N$ points.

 But this step has an extra cost for calculation of each entry in the vector $\mathbf{x}_{j}^{(\theta)}$. It requires $O(K)$ multiplication operations. So this step has the complexity of order $O(\frac{n}{K}(N \log(N) + K))$.
 
\item Step 7 includes $N\frac{n}{K}$ summation operations.
\item Step 8 has the cost of computing Fourier transform of a vector of the length $\frac{n}{K}$ with computation complexity of  $O(\frac{n}{K} \; \log (\frac{n}{K}))$.
\item Step 9 has the cost of evaluating $\mathbf{u}_{p}(z)$, a polynomial of degree $N-1$, at $N$ points with the computation complexity of  $O \left (\frac{n}{K} \left ( N \; \log(N) \right ) \right )$. 
\item Step 10 has the computation complexity of order $O(Kn)$.
\end{itemize}
Steps 4-8 are done by each server. So the computation complexity of each server is of the order $O \left (\frac{n}{K} \log \frac{n}{K} \right )$. Steps 1-3,9,10 are done by the prover. So the computation complexity of the prover is dominated by $O(Kn)$.
\subsection{Main algorithm}
\label{H_algorithm}
In the previous subsection we show how to develop a MPC scheme for the computation of $\mathscr{FFT}_{\mathcal{D}} \left (\mathbf{\mathscr{FFT}_{\mathcal{S}}^{-1} (\mathbf{a})} \right )$, using  $K+T$ servers. In Appendix \ref{Proposed_Algorithm1}, we show how to do the same thing for $\mathbf{\hat{a}} = \mathscr{FFT}_{\mathcal{D}}^{-1}(\mathbf{a})$. We can put them together to develop Algorithm \ref{Main_algorithm}. Note that in this paper we focus on a version of QAP-based zkSNARK, proposed by Groth \cite{groth2016size}. However, this scheme can be adopted for other QAP-based zkSNARKs \cite{ben2017scalable, gennaro2013quadratic, groth2016size, parno2013pinocchio, lipmaa2013succinct, danezis2014square}.

\begin{algorithm}[!htbp]
\setstretch{1.75}
\caption{Multiparty algorithm for proof generation}
\label{Main_algorithm}
\begin{algorithmic}[1]
\Statex
\textbf{Input}: $\mathcal{EK}$, function $F$ and its inputs.
\State \hskip1em
Convert $F$ into arithmetic circuit and build QAP just like the setup phase.
\State \hskip1em
Compute the values of all wires in the arithmetic circuit.
		\State \textbf{function} $\mathsf{polynomial-division}$ ($P(x)$, $T(x)$)
		\State \hskip1em Let $\mathbf{a} = \left [a_j \right ]_{j = 0}^{n-1}$, where $a_j = L(\omega^j)$. 
		\State \hskip1em Let $\mathbf{b} = \left [b_j \right ]_{j = 0}^{n-1}$, where $b_j = R(\omega^j)$. 
		\State \hskip1em Let $\mathbf{c} = \left [c_j \right ]_{j = 0}^{n-1}$, where $c_j = O(\omega^j)$.
		\State \hskip1em Calculate $\mathbf{{a}''} = \mathscr{FFT}_{\mathcal{D}}\left(\mathscr{FFT}_{\mathcal{S}}^{-1}(\mathbf{a}) \right)$ using Algorithm \ref{fft_MPC_v2}. 
		\State \hskip1em Calculate $\mathbf{{b}''} = \mathscr{FFT}_{\mathcal{D}}\left(\mathscr{FFT}_{\mathcal{S}}^{-1}(\mathbf{b}) \right)$ using Algorithm \ref{fft_MPC_v2}.
		\State \hskip1em Calculate $\mathbf{{c}''} = \mathscr{FFT}_{\mathcal{D}}\left(\mathscr{FFT}_{\mathcal{S}}^{-1}(\mathbf{c}) \right)$ using Algorithm \ref{fft_MPC_v2}.
		\State \hskip1em Calculate $\left [h_j \right ]_{j=0}^{n-1}$, where $h_j = ({a_j}''.{b_j}''-{c_j}'') \div T(\eta \omega^{j})$.
		\State \hskip1em Calculate $\left [{h}'_j \right]_{j=0}^{n-1} = \mathscr{FFT}_{\mathcal{D}}^{-1}(\left [h_j \right ]_{j=0}^{n-1})$ using Algorithm \ref{fft1_MPC}.
		\State \textbf{return} $\left [{h}'_j \right]_{j=0}^{n-1}$

		\State \textbf{function} $\mathsf{compute-proof}$
		\State \hskip1em Choose secret parameters $r, q$ independently and uniformly at random from $\mathbb{F}$.
		\State \hskip1em Calculate $\llbracket L_r \rrbracket_1 = \llbracket \alpha \rrbracket_1 + \sum_{i = 0}^{m} W_i \llbracket L_i(s) \rrbracket_1 + r \llbracket \delta \rrbracket_1$ where $W_i$ is the value carried by wire $i$.
		\State \hskip1em Calculate $ \llbracket R_q \rrbracket _1 = \llbracket \beta \rrbracket_1 + \sum_{i = 0}^{m} W_i  \llbracket R_i(s) \rrbracket_1 + q \llbracket \delta \rrbracket_1$
		\State \hskip1em Calculate $ \llbracket R_q \rrbracket_2 = \llbracket \beta \rrbracket_2 + \sum_{i = 0}^{m} W_i \llbracket R_i(s)]_2 + q \llbracket \delta \rrbracket_2$
		\State \hskip1em Calculate $ \llbracket K_{r,q} \rrbracket_1 = q \llbracket L_r \rrbracket_1 + r \llbracket R_q \rrbracket_1 - rq \llbracket \delta \rrbracket_1 + \sum_{i \in \mathcal{I}_{mid}} W_i \llbracket k_i^{pk} \rrbracket_1 + \sum_{j=0}^{n-2} h'_j \llbracket t_j \rrbracket_1$
		\State \textbf{return} $\pi = \{ \llbracket L_r \rrbracket_1, \llbracket R_q \rrbracket_2, \llbracket K_{r,q} \rrbracket_1 \}$
		\Statex
		\textbf{Outputs}: $\pi$, public inputs, and public outputs of $F$.
\end{algorithmic}
\end{algorithm}

\subsubsection{Privacy}
Since in each of the algorithms that we use in the main algorithm we use independent random vectors to share the secret input, one can easily show that even if we put all of these algorithms together, it does not leak any information to any subset of $T$ colluding servers.
\subsubsection{Computation complexity}
In Algorithm \ref{Main_algorithm}, the servers participate in the execution of steps 7, 8, 9, 11. So according to Subsection \ref{Proposed_Algorithm2} and Appendix \ref{Proposed_Algorithm1}, the computation complexity in each server is equal to $O(\frac{n}{K} \; \log (\frac{n}{K}))$. 

The computation complexity of the prover in steps 7, 8, 9, 11 is $O(Kn)$. Recall from Section \ref{Groth_prover} that the computation complexity of the prover in steps 1-6 and 10 is $O(n)$, and in steps 14-18 is $O(m\kappa)$. Assuming that $\kappa$, the security parameter, is not too large, and $m$, the number of the wires in the  arithmetic circuit, be at the same order of $n$, which is usually the case, the term $O(Kn)$ is the dominant term. So we can say that the overall computation complexity of the prover is equal to $O(Kn)$.



\section{Discussion and conclusion}
\label{Discussion and conclusion}
 Zero knowledge proofs are fundamental tools with a wide range of applications. Despite extensive efforts dedicated to optimize zero knowledge proof algorithms, still they incur a lot of computation complexities. 
In this paper we have presented a secure multi party algorithm to delegate the task of prover to several servers, where servers are untrusted and have limited computation resources. We have focused on QAP-based zkSNARKs due to its importance in practice, however a similar approach can be taken to distribute other kinds of zero knowledge proof systems.

\bibliographystyle{ieeetr}
\bibliography{Bibliography}
\begin{appendices}
\clearpage
\section{The multiparty algorithm for computing $\mathbf{\hat{a}} = \mathscr{FFT}_{\mathcal{D}}^{-1}(\mathbf{a})$}
\label{Proposed_Algorithm1}
Suppose that the prover has a large secret vector $\mathbf{a}$ of dimension $n$, and aims to compute $\mathscr{FFT}_{\mathcal{D}}^{-1}(\mathbf{a})$, using a cluster of $N$ semi-honest servers, where up to $T$ of them may collude. Here we propose Algorithm \ref{fft1_MPC} in which the computation complexity of each server is equal to $O(\frac{n}{K} \log (\frac{n}{K}))$, and the computation complexity of the prover is equal to $O(Kn)$.
\begin{algorithm}[!htbp]
\setstretch{1.75}
\caption{Multi party algorithm for computing $\mathbf{\hat{a}} = \mathscr{FFT}_{\mathcal{D}}^{-1}(\mathbf{a})$}
\label{fft1_MPC}
\begin{algorithmic}[1]
\Statex
\textbf{Input}: vector $\mathbf{a}=\left [a_i \right ]_{i=0}^{n-1}$ of the length $n$.
\Statex
Prover does the following steps:
\State \hskip1em 
partitions $\mathbf{a}$ into $K$ vectors $\mathbf{a}^{(j)}=\left [a_{Kt+j-1} \right ]_{t=0}^{\frac{n}{K}-1}$ for $j \in \{1,\dots,K\}$.
\State \hskip1em 
picks $T$ vectors $\mathbf{v}_j \in \mathbb{F}$, $j \in \{K+1,\dots,K+T\}$, independently and uniformly at random from ${\frac{n}{K}}$.
\State \hskip1em 
forms $\mathbf{u}(z)$ according to (\ref{FFT_first_algorthm_u}), and sends $\mathbf{u}(\alpha_\theta)$ to Server $\theta$, where $\theta \in \{1,\dots,N\}$.
\Statex
Server $\theta$ does the followings:
\State \hskip1em
computes $\mathscr{FFT}_{\mathcal{S'}}^{-1}(\mathbf{u}(\alpha_\theta))$ for $\mathcal{S'} = \left \{1, \omega^K, \dots, \omega^{K \left (\frac{n}{K}-1 \right)} \right \}$, and sends it to the prover.
\Statex
Upon receiving $\mathscr{FFT}_{\mathcal{S'}}^{-1}(\mathbf{u}(\alpha_\gamma))$, $\gamma \in \{1,\dots, N \}$, the prover does the following steps:
\State \hskip 1em 
forms $\mathbf{u}''(z)$ as defined in (\ref{One_FFT_decode}), and evaluates $\mathbf{u}''(\beta_j)$ for $j \in \{1, \dots, K\}$.
\State \hskip1em
calculates the vector $\mathbf{\hat{a}} =\left [\hat{a}_i \right ]_{i=0}^{n-1}$ as $\hat{a}_i = \frac{K\eta^{-i}}{n} \sum_{j = 1}^{K} \omega^{-i(j-1)} \left (\mathbf{u}''(\beta_j) \right )_{i \mod \frac{n}{K}}$.
\Statex
\textbf{Output}: $\mathbf{\hat{a}}$
\end{algorithmic}
\end{algorithm}

Let define the polynomial $\mathbf{u}''(z)$ as,
\begin{align} \label{One_FFT_decode}
\mathbf{u}''(z) \triangleq \sum_{\gamma = 1}^{N} \mathscr{FFT}_{\mathcal{S'}}^{-1}(\mathbf{u}(\alpha_\gamma)) \prod_{k = 1, k \neq \gamma}^{N} \frac{z - \alpha_{k}}{\alpha_{\gamma} - \alpha_{k}},
\end{align}
which is of degree $N-1 = K + T -1$. In step 5 of Algorithm \ref{fft1_MPC}, the prover needs to evaluate it in $K$ different points $\beta_j$ for $j \in \{1,\dots,K \}$. 

\end{appendices}
\end{document}